\documentclass{article}

\usepackage{varioref}
\usepackage{subfigure}
\usepackage{srcltx} 
\usepackage{epsfig}
\usepackage[T1]{fontenc}
\usepackage[utf8]{inputenc}
\usepackage[english]{babel}
\usepackage{times}
\usepackage{color}
\usepackage{amsmath,enumerate}
\usepackage{amssymb}
\usepackage{amscd}
\usepackage{latexsym}
\usepackage{enumerate}
\usepackage{tabularx}
\usepackage{stmaryrd}
\usepackage{vmargin}
\usepackage[ruled,vlined]{algorithm2e}
\usepackage{mathrsfs}
\usepackage{a4}
\usepackage{url}
\usepackage{hyperref}
\usepackage[affil-it]{authblk}

\setmarginsrb{3.5cm}{2cm}{3.5cm}{2cm}{1cm}{1cm}{1cm}{1cm}

\title{Entropy compression method applied to graph
  colorings\thanks{This research is partially supported by the ANR
    EGOS, under contract ANR-12-JS02-002-01.}}

\author[a]{Daniel Gonçalves}
\author[b]{Mickaël Montassier}
\author[c]{Alexandre Pinlou}

\affil[a]{{\small CNRS, LIRMM}}
\affil[b]{{\small University of Montpellier, LIRMM}}
\affil[c]{{\small University Paul-Valéry Montpellier, LIRMM\medskip}}
\affil[ ]{{\small 161 rue Ada, 34095 Montpellier Cedex 5, France}}
\affil[ ]{\rm\footnotesize\{\url{daniel.goncalves, mickael.montassier, alexandre.pinlou}\}\url{@lirmm.fr}}

\date{January 22, 2014}

\newenvironment{proof}{\par \noindent \textbf{Proof.} }{\hfill$\Box$\medskip}
\newenvironment{proofof}[1]{\par \noindent \textbf{Proof of #1.}}{\hfill$\Box$\medskip} 

\newtheorem{theorem}{Theorem}
\newtheorem{corollary}[theorem]{Corollary}
\newtheorem{definition}[theorem]{Definition}

\newtheorem{lemma}[theorem]{Lemma}

\newtheorem{remark}[theorem]{Remark}
\newtheorem{claim}[theorem]{Claim}

\newcommand{\smallO}[1]{\ensuremath{\mathop{}\mathopen{}o\mathopen{}\left(#1\right)}}

\def\paren#1{\left( #1 \right)}

\def\floor#1{\left\lfloor #1 \right\rfloor}
\def\ceil#1{\left\lceil #1 \right\rceil}

\newcommand{\set}[1]{\left\{#1\right\}}

\newcommand{\ac}{\textsc{AcyclicColoring\_G}\xspace}
\newcommand{\acb}{\textsc{AcyclicColoring-V2\_G}\xspace}
\newcommand{\acm}{\textsc{AcyclicColoringGamma\_G}\xspace}
\newcommand{\gc}{\textsc{Coloring\_G}\xspace}
\newcommand{\algo}[1]{"\texttt{#1}"}
\newcommand{\nbmaxcolmadgamma}[1]{ 1+ \Delta \left(1+\sqrt{2#1+4}  \right)}

\newcommand{\nbmaxcolgamma}{\nbmaxcolmadgamma{\gamma}}
\newcommand{\nbmaxcoldelta}{\frac32\Delta^\frac43 + 5 \Delta -14}
\newcommand{\nbmaxcoldeltabis}{\frac32\Delta^\frac43 + \Delta +\frac{8\Delta^\frac43}{\Delta^\frac23-4}+1}

\newcommand{\ovphi}{\overline{\varphi}}

\newcommand{\jset}{{j\cdot{\rm Set}}}
\DeclareMathOperator{\type}{\mathscr{T}}
\DeclareMathOperator{\class}{\mathscr{C}}

\DeclareMathOperator{\NUV}{\tt NextUncoloredElement}
\DeclareMathOperator{\ONBE}{\tt BadEventClass}
\DeclareMathOperator{\USBE}{\tt UncolorSetBadEvent}
\DeclareMathOperator{\RBE}{\tt RecoverBadEvent}

\DeclareMathOperator{\BET}{\tt BadEventType}

\begin{document}

\maketitle

\begin{abstract}
  Based on the algorithmic proof of Lov\'asz local lemma due to Moser
  and Tardos, the works of Grytczuk {\em et al.} on words, and
  Dujmovi\'c {\em et al.} on colorings, Esperet and
  Parreau developed a framework to prove upper bounds for several
  chromatic numbers (in particular acyclic chromatic index, star
  chromatic number and Thue chromatic number) using the so-called
  \emph{entropy compression method}.
  
  Inspired by this work, we propose a more general framework and a
  better analysis. This leads to improved upper bounds on chromatic
  numbers and indices. In particular, every graph with maximum degree
  $\Delta$ has an acyclic chromatic number at most
  $\frac{3}{2}\Delta^{\frac43} + O(\Delta)$.  Also every planar graph
  with maximum degree $\Delta$ has a facial Thue choice number at most
  $\Delta + O(\Delta^\frac 12)$ and facial Thue choice index at most
  $10$.
\end{abstract}

\section{Introduction}

In the 70's, Lov\'asz introduced the celebrated \emph{Lov\'asz Local
  Lemma} (LLL for short) to prove results on 3-chromatic
hypergraphs~\cite{EL75}.  It is a powerful probabilistic method to
prove the existence of combinatorial objects satisfying a set of
constraints.  Since then, this lemma has been used in many occasions.
In particular, it is a very efficient tool in graph coloring to
provide upper bounds on several chromatic
numbers~\cite{AGHR02,ASZ01,FRR04,GP05,Hat05,HV+08,MR97,MR98}. Recently
Moser and Tardos~\cite{MT10} designed an algorithmic version of LLL by
means of the so-called {\em Entropy Compression Method}. This method
seems to be applicable whenever LLL is, with the benefits of providing
tighter bounds. Using ideas of Moser and Tardos~\cite{MT10}, Grytczuk
{\em et al.}~\cite{GKM13} proposed new approaches in the old field of
nonrepetitive sequences. Inspired by these works, Dujmovik {\em et
  al}~\cite{DJ+12} gave a first application of the entropy compression
method in the area of graph colorings (on Thue vertex coloring and
some of its game variants). As the approach seems to be extendable to
several graph coloring problems, Esperet and Parreau~\cite{EP12}
developed a general framework and applied it to acyclic edge-coloring,
star vertex-coloring, Thue vertex-coloring, each time improving the
best known upper bound or giving very short proofs of known bounds.
In the continuity of these works, we provide a more general method and
give new tools to improve the analysis. As application of that method,
we obtain some new upper bounds on some invariants of graphs, such as
acyclic choice number, facial Thue chromatic number/index, ...

The paper is organized as follows. In Section~\ref{sec:acyclic}, we
present the method and apply it to acyclic vertex coloring. It will be
the occasion of providing improved bounds (in terms of the maximum
degree). Then, in Sections~\ref{sec:general_la_vraie} and
\ref{sec:general}, we describe the general method and provide its
analysis. Finally, Section~\ref{sec:application} is dedicated to
the applications of that method.

% Local variables:
% mode: latex
% TeX-master: "0_main.tex"
% End:

         %% section 1
\section{Acyclic coloring of graphs}\label{sec:acyclic}

A {\em proper coloring} of a graph is an assignment of colors to the
vertices of the graph such that two adjacent vertices do not use the
same color. A {\em $k$-coloring} of a graph $G$ is a proper coloring
of $G$ using $k$ colors ; a graph admitting a $k$-coloring is said to
be {\em $k$-colorable}. An {\em acyclic coloring} of a graph $G$ is a
proper coloring of $G$ such that $G$ contains no bicolored cycles ; in
other words, the graph induced by every two color classes is a forest.
Let $\chi_a(G)$, called the \emph{acyclic chromatic number}, be the
smallest integer $k$ such that the graph $G$ admits an acyclic
$k$-coloring.

Acyclic coloring was introduced by Grünbaum \cite{Gru73}. In
particular, he proved that if the maximum degree $\Delta$ of $G$ is at
most $3$, then $\chi_a(G) \le 4$. Acyclic coloring of graphs with
small maximum degree has been extensively
studied~\cite{Bur79,DHN10,FR08,Fie13, H11,KS11, YS09,YS09+,YS11} and the
current knowledge is that graphs with maximum degree $\Delta \le 4,5$,
and $6$, respectively verify $\chi_a(G) \le 5,7$, and
$11$~\cite{Bur79,KS11,H11}. For higher values of the maximum degree, Kostochka
and Stocker~\cite{KS11} showed that $\chi_a(G) \le 1 + \left
  \lfloor\frac{(\Delta+1)^2}{4}\right \rfloor$.  Finally, for large
values of the maximum degree, Alon, McDiarmid, and Reed \cite{AMR91}
used LLL to prove that every graph with maximum degree $\Delta$
satisfies $\chi_a(G)\le \left\lceil 50 \Delta^{4/3} \right\rceil$.
Moreover they proved that there exist graphs with maximum degree
$\Delta$ for which 
$\chi_a=\Omega\left(\frac{\Delta^{4/3}}{(\log \Delta)^{4/3}}\right)$.  Recently, the upper bound
was improved to $\left \lceil6.59\Delta^{\frac43} + 3.3\Delta
\right\rceil$ by Ndreca et al.~\cite{NPS12} and then to $2.835
\Delta^{\frac43} + \Delta$ by Sereni and Volec~\cite{SV13}.
  
We improve this upper bound (for large $\Delta$) by a constant factor.
\begin{theorem}\label{thm:acy-delta}
Every graph $G$ with maximum degree $\Delta\ge 24$ is such that 
$$\chi_a(G) <\min\set{\nbmaxcoldelta,\quad\nbmaxcoldeltabis}.$$
\end{theorem}
At the end of
Section \ref{subsubsec:algo} (see Remark \ref{re:eminem}), we give a
method to refine these upper bounds, improving on Kostochka and Stocker's
bound as soon as $\Delta\ge 27$.

Alon, McDiarmid, and Reed~\cite{AMR91} also considered the acyclic
chromatic number of graphs having no copy of $K_{2,\gamma +1}$ (the
complete bipartite graph with partite sets of size 2 and $\gamma+1$)
in which the two vertices in the first class are non-adjacent. Let
${\cal K}_\gamma$ be the familly of such graphs. Such structure 
contains many cycles of length~$4$ and they are an obstruction to get
an upper bound on the acyclic chromatic number linear in $\Delta$. Again using
LLL, they proved that every graph $G\in {\cal K}_\gamma$ with maximum
degree $\Delta$ satisfies $\chi_a(G) \le \lceil 32 \sqrt{\gamma}
\Delta\rceil$. 

Using similar techniques as
for Theorem~\ref{thm:acy-delta}, we obtain:
\begin{theorem}\label{thm:gamma}
  Let $\gamma\ge 1$ be an integer and $G\in {\cal K}_\gamma$ with
  maximum degree $\Delta$. We have
  $\chi_a(G) \le  \nbmaxcolgamma$.
\end{theorem}

As it is simpler, let us start with the proof of
Theorem~\ref{thm:gamma} that will serve as an educational example of
the entropy compression method.

% Local variables:
% mode: latex
% TeX-master: "0_main.tex"
% End:
       %% section 2
\subsection{Graphs with restrictions on $K_{2,\gamma+1}$'s}\label{ssec:acyclic-gamma}

We prove Theorem~\ref{thm:gamma} by contradiction. Suppose there
exists a graph $G \in {\cal K}_\gamma$ with maximum degree $\Delta$
such that $\chi_a(G) > \nbmaxcolgamma$. We define an algorithm that
``tries'' to acyclically color $G$ with $\kappa = \nbmaxcolgamma$
colors. Define a total order $\prec$ on the vertices of $G$.

%%%%%%%%%%%%%%%%%%%%%%%%%%%%%%%%%%%%%%%%%%%%%%%%%%%%%%%%%%%%%%%%%%%%%% 
%%%%%%%%%%%%%%%%%%%%%%%%%%%%%%%%%%%%%%%%%%%%%%%%%%%%%%%%%%%%%%%%%%%%%% 
%% ALGORITHM                              %%
%%%%%%%%%%%%%%%%%%%%%%%%%%%%%%%%%%%%%%%%%%%%%%%%%%%%%%%%%%%%%%%%%%%%%% 
%%%%%%%%%%%%%%%%%%%%%%%%%%%%%%%%%%%%%%%%%%%%%%%%%%%%%%%%%%%%%%%%%%%%%% 
\subsubsection{The algorithm}

Let $V\in\set{1,2,\ldots,\kappa}^t$ be a vector of length $t$, for some
arbitrarily large $t \gg n=|V(G)|$. Algorithm~\acm (see
\vpageref[below]{algo:acm}) takes the 
vector $V$ as input and returns a partial acyclic coloring $\varphi :
V(G) \to \set{\bullet,1,2,\ldots,\kappa}$ of $G$ ($\bullet$ means that
the vertex is uncolored) and a text file $R$ that is called a
\emph{record} in the remaining of the paper. The acyclic coloring
$\varphi$ is necessarily partial since we try to color $G$ with a
number of colors less than its acyclic chromatic number.
For a given vertex $v$ of $G$, we denote by $N(v)$ the set of
neighbors of $v$.

\begin{algorithm}
  \DontPrintSemicolon
  \LinesNumbered
  
  \SetKwInOut{Input}{Input}
  \SetKwInOut{Output}{Output}
  \Input{$V$ (vector of length $t$).}
  \Output{($\varphi$, $R$).\vspace{.2cm}}
  
  \For{ all $v$ in $V(G)$}{
    $\varphi(v)\gets \bullet$ 
  }
  $R\gets newfile()$

  \For{$i\gets 1$ \KwTo $t$}{

    Let $v$ be the smallest (w.r.t. $\prec$) uncolored vertex of $G$
    
    $\varphi(v) \gets V[i]$
    
    Write "\texttt{Color $\backslash$n}" in $R$
    
    \If{$\varphi(v)=\varphi(u)$ \rm{for} $u\in N(v)$} %% proper
    {
      \tcp{Proper coloring issue}
      $\varphi(v) \gets \bullet$
      
      Write "\texttt{Uncolor, neighbor $u\ \backslash$n}" in $R$
      
    }
    \ElseIf{$v$ belongs to a bicolored cycle of length $2k$ ($k\ge 2$), say $(v=u_1,\ldots,u_{2k})$} %% C_4
    {
      \tcp{Bicolored cycle issue}
      
      \For{$j\gets 1$ \KwTo $2k-2$}{ $\varphi(u_j) \gets \bullet$ }
      Write "\texttt{Uncolor, $2k$-cycle $(v=u_1,\ldots,u_{2k})\ 
        \backslash$n}" in $R$ } } \Return ($\varphi$, $R$)
  \caption{\acm\label{algo:acm}}
\end{algorithm}

Algorithm \acm runs as follows. Let $\varphi_i$ be the partial
coloring of $G$ after $i$ steps (at the end of the $i^{\rm th}$ loop).
At Step $i$, we first consider $\varphi_{i-1}$ and we color the
smallest uncolored vertex $v$ with $V[i]$ (line 6 of
the algorithm).  We then verify whether one of the 
following types bad events happens:

\begin{enumerate}[Event 1:]
\item $G$ contains a monochromatic edge $vu$ for some $u$ (line 8 of the
  algorithm) ;
\item[Event $k$:] $G$ contains a bicolored cycle of length $2k$
  $(v=u_1,u_2,\ldots,u_{2k})$ (line 11 of the algorithm).
\end{enumerate}
If such events happen, then we uncolor some vertices (including $v$) in
order that none of the two previous events remains. 
Clearly, $\varphi_i$ is a partial acyclic coloring of $G$. Indeed,
since Event~1 is avoided, $\varphi_i$ is a proper coloring and since
Event~2 is avoided, $\varphi_i$ is acyclic.

\bigskip
\begin{proofof}{Theorem~\ref{thm:gamma}}
  Let us first note that the function defined by Algorithm \acm is
  injective. This comes from the fact that from each output of the
  algorithm, one can determine the corresponding input by
  Lemma~\ref{lem:vo2}.
  Now we obtain a contradiction by showing that the number of possible
  outputs is strictly smaller than the number of possible inputs when
  $t$ is chosen large enough. The number of possible
  inputs is exactly $\kappa^t$ while the number of possible outputs is
  $\smallO{\kappa^t}$, as it is at most $(1+\kappa)^n\times
  \smallO{\kappa^t}$. Indeed, there are at most $(1+\kappa)^n$ possible
  partial $\kappa$-colorings of $G$ and there are at most
  $\smallO{\kappa^t}$ possible records by Lemma~\ref{lem:r2}. Therefore,
  assuming the existence of a counterexample $G$ leads us to a
  contradiction. That concludes the proof of Theorem~\ref{thm:gamma}.
\end{proofof}

%%%%%%%%%%%%%%%%%%%%%%%%%%%%%%%%%%%%%%%%%%%%%%%%%%%%%%%%%%%%%%%%%%%%%% 
%%%%%%%%%%%%%%%%%%%%%%%%%%%%%%%%%%%%%%%%%%%%%%%%%%%%%%%%%%%%%%%%%%%%%% 
%% ANALYSIS                              %%
%%%%%%%%%%%%%%%%%%%%%%%%%%%%%%%%%%%%%%%%%%%%%%%%%%%%%%%%%%%%%%%%%%%%%% 
%%%%%%%%%%%%%%%%%%%%%%%%%%%%%%%%%%%%%%%%%%%%%%%%%%%%%%%%%%%%%%%%%%%%%% 

\subsubsection{Algorithm analysis}

Recall that $\varphi_i$ denotes the partial acyclic coloring obtained
after $i$ steps. Let us denote by $\ovphi_i \subset V(G)$ the set of
vertices that are colored in $\varphi_i$. Let also $v_i$, $R_i$ and
$V_i$ respectively denote the current vertex $v$ of the $i^{\rm th}$
step, the record $R$ after $i$ steps, and the input vector $V$
restricted to its $i$ first elements.  Observe that as $\varphi_i$ is
a partial acyclic $\kappa$-coloring of $G$, and as $G$ is not acyclically
$\kappa$-colorable, we have that $\ovphi_i \subsetneq V(G)$, and thus 
$v_{i+1}$ is well defined. This also implies that $R$ has $t$
\algo{Color} lines. Finally observe that $R_i$ corresponds to the lines
of $R$ before the $(i+1)^{\rm th}$ \algo{Color} line.

\begin{lemma}\label{lem:vo2}
  One can recover $V_i$ from $(\varphi_i,R_i)$.
\end{lemma}

\begin{proof} By induction on $i$. Trivially, $V_0$ (which is empty) 
  can be recovered from $(\varphi_0,R_0)$.  Consider now
  $(\varphi_i,R_i)$ and let us try to recover $V_i$. It is thus
  sufficient to recover $R_{i-1}$, $\varphi_{i-1}$, and $V[i]$.  As
  observed before, to recover $R_{i-1}$ from $R_i$ it is sufficient to
  consider the lines before the last (i.e. the $i^{\rm th}$)
  \algo{Color} line.  Then reading $R_{i-1}$, one can easily recover
  $\ovphi_{i-1}$ and deduce $v_i$.  Note that in the $i^{\rm th}$ step
  we wrote one or two lines in the record: exactly one \algo{Color}
  line followed by either nothing, or one \algo{Uncolor, neighbor}
  line, or one \algo{Uncolor, $2k$-cycle} line. Indeed there cannot be an
  \algo{Uncolor, $2k$-cycle} line following an \algo{Uncolor, neighbor}
  line, as $v$ would be uncolored by the algorithm before considering
  bicolored cycles passing through $v$. Let us
  consider these three cases separately.

  \begin{itemize}  
  \item If Step $i$ was a color step alone, then $V[i]
    = \varphi_i(v_i)$ and $\varphi_{i-1}$ is obtained from $\varphi_i$
    by uncoloring $v_i$.
    
  \item If the last line of $R_i$ is \algo{Uncolor, neighbor $u$},
    then $V[i] = \varphi_i(u)$ and $\varphi_{i-1} = \varphi_i$.

  \item If the last line of $R_i$ is \algo{Uncolor, $2k$-cycle $(u_1,\ldots
      ,u_{2k})$}, then $V[i]=\varphi_i(u_{2k-1})$ and $\varphi_{i-1}$ is
    obtained from $\varphi_i$ by coloring the vertices $u_j$ for $2\le
    j\le 2k-2$ (which were uncolored in $\varphi_i$), in such a way
    that $\varphi_{i-1}(u_j)$ equals $\varphi_i(u_{2k-1})$ if $j \equiv
    1\mod 2$, or equals $\varphi_i(u_{2k})$ otherwise. Note that this is
    possible because in the $i^{\rm th}$ loop, the algorithm uncolored
    neither $u_{2k-1}$ nor $u_{2k}$.

  \end{itemize} 
  This concludes the proof of the lemma.
\end{proof}

Let us now bound the number of possible records.

\begin{lemma}\label{lem:r2} 
  Algorithm~\acm produces at most $\smallO{\kappa^t}$ distinct
  records~$R$.
\end{lemma}

\begin{proof}
  Since Algorithm \acm fails to color $G$, the record $R$ has exactly
  $t$ \algo{Color} lines (i.e. the algorithm consumes the whole input
  vector). It contains also \algo{Uncolor} lines of different types:
  \algo{neighbor} (type 1), \algo{$4$-cycle} (type 2),
  \algo{$6$-cycle} (type 3), \dots \algo{$n$-cycle} (type $\frac
  n2$).
  Let $\type=\set{1,2,\ldots,\frac n2}$ be the set of bad event types.
  Let denote $s_j$ the number of uncolored vertices when a bad event
  of type $j$ occurs. Observe that:
  \begin{itemize}
  \item For every \algo{Uncolor, neighbor} step, the algorithm uncolors
    1 previously colored vertex. Hence set $s_1=1$.
  \item For every \algo{Uncolor, $2k$-cycle} step, where the cycle has length
    $2k$, the algorithm uncolors $2k-2$ previously colored vertices. Hence 
    set $s_k=2k-2$ for $2\le k \le {\floor {n/2}}$.
  \end{itemize}

  To compute the total number of possible records, let us compute how
  many different entries, denoted $C_j$, an \algo{Uncolor} step of type
  $j$ can produce in the record.  Observe that:
  \begin{itemize}

  \item An \algo{Uncolor, neighbor} line can produce $\Delta$ different
    entries in the record, according to the neighbor of $v$ (the
    vertex just
    colored by the algorithm) that shares the same color. Hence set $C_1=\Delta$.

  \item An \algo{Uncolor, $2k$-cycle} line involving a cycle of length $2k$
    can produce as many different entries in the record as the number of
    $2k$-cycles going through $v$. Thus this number of entries is at
    most $\frac12 \gamma \Delta^{2k-2}$ according to Lemma 3.2 of
    \cite{AMR91}. Hence set $C_k=\frac12 \gamma \Delta^{2k-2}$ for $2\le
    k \le {\floor {n/2}}$.
  \end{itemize}

  We complete the proof by means of
  Theorem~\ref{thm:nb_of_rec} of Section~\ref{sec:general} (see
  \vpageref[below]{thm:nb_of_rec}). Theorem~\ref{thm:nb_of_rec} applies
  on Algorithm~\gc which is a generic version of Algorithm~\acm.
  Consequently, let us consider the following polynomial $Q(x)$:
  \begin{eqnarray*}
    Q(x) &=& 1+ \sum_{i\in\type} C_i x^{s_i}\\
         &=&1+\Delta x + \sum_{2\le i\le\frac n2}\frac12 \gamma \Delta^{2i-2}x^{2i-2}\\
         &<&1+\Delta x + \frac{\gamma\Delta^2x^2}{2-2\Delta^2x^2} \qquad\qquad \;
             \mbox{for} \;
             x<\frac 1\Delta
  \end{eqnarray*}
  Setting $X =
  \frac{1}{\Delta}\sqrt{\frac{2}{\gamma+2}}$, we have:
  $$
  \frac{Q(X)}{X}  < \Delta \sqrt{\frac{\gamma+2}{2}}\left(1+ \sqrt{\frac{2}{\gamma+2}}
    +1\right) = \Delta \left(1+\sqrt{2\gamma+4} \right) \le \kappa
  $$ 

  Since $\gamma\ge 1$, then $\frac2{\gamma+2}<1$ and thus we have $0< X <
  \frac1\Delta \le 1$. Therefore, Algorithm \acm produces at most
  $o(\kappa^t)$ different records by Theorem~\ref{thm:nb_of_rec}.  This completes the proof.
\end{proof}

% Local variables:
% mode: latex
% TeX-master: "0_main.tex"
% ispell-local-dictionary: "english"
% End:

 %% section 2.1
\subsection{Graphs with maximum degree $\Delta$}
\label{ssec:acyclic}

To prove Theorem~\ref{thm:acy-delta}, we prove that, given a graph $G$
with maximum degree $\Delta$, we have $\chi_a(G)< \nbmaxcoldelta$ for
$\Delta\ge 24$ in Section~\ref{subsubsec:algo} and that $\chi_a(G) <
\nbmaxcoldeltabis$ for $\Delta\ge 9$ in
Section~\ref{subsubsec:algobis}.

\begin{figure}
  \begin{center}
    \includegraphics[scale=0.65]{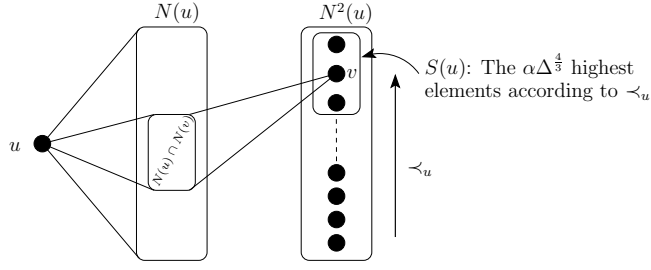}
    \caption{Example of a special couple $(u,v)$.\label{fig:special}}
  \end{center}
\end{figure}

The proof is made by contradiction. Suppose there exists a graph $G$ with maximum degree $\Delta$ which is
a counterexample to Theorem~\ref{thm:acy-delta}.
Define a total order $\prec$
on the vertices of $G$. Let $N(u)$ and $N^2(u)$ be respectively the
set of neighbors and distance-two vertices of $u$.  For each pair of
non-adjacent vertices $u$ and $v$, let $N(u,v) = N(u) \cap N(v)$, and let
$\deg(u,v) = |N(u,v)|$.  For each vertex $u$ of $G$, let the order
$\prec_u$ on $N^2(u)$ be such that $v\prec_u w$ if $\deg(u,v) <
\deg(u,w)$, or if $\deg(u,v) = \deg(u,w)$ but $v\prec w$.  A couple of
vertices $(u,v)$ with $v\in N^2(u)$ is {\em special} if there are less
than $\alpha\Delta^{\frac{4}{3}}$ ($\alpha$ is a constant to be set later)
vertices $w$ such that $v\prec_u w$. %; see Figure~\ref{fig:special}. 
That is, $(u,v)$ is special if and only if, $v$ is in the
$\alpha\Delta^{4/3}$ highest elements of $\prec_u$ (see Figure~\ref{fig:special}).  Note that the couple
$(u,v)$ may be special while the couple $(v,u)$ may be non-special.
Let us denote $S(u)\subseteq N^2(u)$ the set of vertices $v$ such that
$(u,v)$ is special.  By
definition, $|S(u)| = \min \left\{ \alpha \Delta^{\frac{4}{3}} , |N^2(u)| \right\}$.

\subsubsection{First upper bound}\label{subsubsec:algo}
By contradiction hypothesis, $\chi_a(G)\ge \frac32\Delta^\frac43 + 5
\Delta - 14$.  Let $\kappa$ be the unique integer such that
$\frac32\Delta^\frac43 + 5 \Delta - 15 \le \kappa <
\frac32\Delta^\frac43 + 5 \Delta - 14$ (i.e. $\kappa = \ceil{
  \frac32\Delta^\frac43 + 5 \Delta - 15 }$).

\paragraph{The algorithm}\ \smallskip

\noindent Let $V\in\set{1,2,\ldots,\kappa}^t$ be a vector of length
$t$. Algorithm~\ac (see
\vpageref[below]{algo:ac}) takes the vector $V$ as input and returns a 
partial acyclic coloring $\varphi : V(G) \to
\set{\bullet,1,2,\ldots,\kappa}$ of $G$ (recall that $\bullet$ means
that the vertex is uncolored) and a record $R$.

%%%%%%%%%%%%%%%%%%%%%%%%%%%%%%%%%%%%%%%%%%%%%%%%%%%%%%%%%%%%%%%%%%%%%%
%%%%%%%%%%%%%%%%%%%%%%%%%%%%%%%%%%%%%%%%%%%%%%%%%%%%%%%%%%%%%%%%%%%%%%
%%                           ALGORITHM  (P_6)                       %%
%%%%%%%%%%%%%%%%%%%%%%%%%%%%%%%%%%%%%%%%%%%%%%%%%%%%%%%%%%%%%%%%%%%%%%
%%%%%%%%%%%%%%%%%%%%%%%%%%%%%%%%%%%%%%%%%%%%%%%%%%%%%%%%%%%%%%%%%%%%%%

\begin{algorithm}[t]
  \DontPrintSemicolon
  \LinesNumbered
  
  \SetKwInOut{Input}{Input}
  \SetKwInOut{Output}{Output}
  \Input{$V$ (vector of length $t$).}
  \Output{($\varphi$, $R$).\vspace{.2cm}}
  
  \For{all $v$ in $V(G)$}{
    $\varphi(v)\gets \bullet$ 
  }
  $R\gets newfile()$

  \For{$i\gets 1$ \KwTo $t$}{

    Let $v$ be the smallest (w.r.t. $\prec$) uncolored vertex of $G$
    
    $\varphi(v) \gets V[i]$
    
    Write "\texttt{Color} $\backslash$n" in $R$
    
    \If{$\varphi(v)=\varphi(u)$ \rm{for} $u\in N(v)$} %% proper
    {
      \tcp{Proper coloring issue}
      $\varphi(v) \gets \bullet$
      
      Write "\texttt{Uncolor, neighbor $u$ $\backslash$n}" in $R$
      
    }
    \ElseIf{$\varphi(v)=\varphi(u)$ \rm{for} $u\in S(v)$}  %% special couple
    {
      \tcp{Special couple issue}
      $\varphi(v) \gets \bullet$
      
      Write "\texttt{Uncolor, special $u$} $\backslash$n" in $R$
    }
    \ElseIf{$v$ belongs to a bicolored cycle of length $4$ $(v=u_1,u_2,u_3,u_4)$}  %% C_4
    {
      \tcp{Bicolored cycle issue}
      $\varphi(v) \gets \bullet$

      $\varphi(u_2) \gets \bullet$
      
      Write "\texttt{Uncolor, cycle $(u_1,u_2,u_3,u_4)$} $\backslash$n" in $R$
    }
    \ElseIf{$v$ belongs to a bicolored path of length 6
      $(u_1,u_2=v,u_3,u_4,u_5,u_6)$ with $u_1\prec u_3$} %% P_6
    { 
      \tcp{Bicolored path issue}
      $\varphi(u_1) \gets \bullet$

      $\varphi(v) \gets \bullet$

      $\varphi(u_3) \gets \bullet$

      $\varphi(u_4) \gets \bullet$

      Write "\texttt{Uncolor, path $(u_1,u_2,u_3,u_4,u_5,u_6)$} $\backslash$n" in $R$
    }
  }
  \Return ($\varphi$, $R$)
  \caption{\ac\label{algo:ac}}
\end{algorithm} 

Algorithm \ac runs as follows. Let $\varphi_i$ be the partial coloring
of $G$ after $i$ steps (at the end of the $i^{\rm th}$ loop).  At Step
$i$, we first consider $\varphi_{i-1}$ and we color the smallest
uncolored vertex $v$ with $V[i]$ (line 6 of the algorithm).
We then verify whether one of the following types of bad events happens:

\begin{enumerate}[Event $N$ (for neighbor):] 
\item[Event $N$ (for neighbor):] $G$ contains a monochromatic edge $vu$ for some $u$ (line 8 of
  the algorithm);
\item[Event $S$ (for special):] $G$ contains a special couple $(v,u)$ with
  $u$ and $v$ having the same color (line 11 of the algorithm);
\item[Event $C$ (for cycle):] $G$ contains a bicolored cycle of length 4 $(v=u_1,u_2,u_3,u_4)$
  (line 14 of the algorithm);
\item[Event $P$ (for path):] $G$ contains a bicolored path of length 6
  $(u_1,u_2=v,u_3,u_4,u_5,u_6)$ with $u_1\prec u_3$ (line 18 of
  the algorithm).
\end{enumerate}
If such events happen, then we modify the coloring (i.e. we
uncolor some vertices as mentioned in Algorithm~\ac) in
order that none of the four previous events remains. Note that at some
Step $i$, for $u$ and $v$ two vertices of $G$ such that $(u,v)$ is a
special couple but $(v,u)$ is not, we may have
$\varphi(u)=\varphi(v)$; this means that $u$ has been colored before
$v$.
Clearly, $\varphi_i$ is a partial acyclic coloring of $G$. Indeed,
since Event~1 is avoided, $\varphi_i$ is a proper coloring ; since
Events~3 and~4 are avoided, $\varphi_i$ is acyclic.

\bigskip
\begin{proofof}{Theorem \ref{thm:acy-delta}}
As in the proof of Theorem \ref{thm:gamma}, we prove that the function
defined by \ac is injective (see Lemma \ref{lem:vo}).  A contradiction
is then obtained by showing that the number of possible outputs is
strictly smaller than the number of possible inputs when $t$ is chosen
large enough compared to $n$. The number of possible inputs is exactly
$\kappa^t$ while the number of possible outputs is $\smallO{\kappa^t}$,
as the number of possible $(1+\kappa)$-colorings of $G$ is
$(1+\kappa)^n$ and the number of possible records is $\smallO{\kappa^t}$
(see Lemma~\ref{lem:r}).
\end{proofof}

%%%%%%%%%%%%%%%%%%%%%%%%%%%%%%%%%%%%%%%%%%%%%%%%%%%%%%%%%%%%%%%%%%%%%%
%%%%%%%%%%%%%%%%%%%%%%%%%%%%%%%%%%%%%%%%%%%%%%%%%%%%%%%%%%%%%%%%%%%%%%
%%                            ANALYSIS                              %%
%%%%%%%%%%%%%%%%%%%%%%%%%%%%%%%%%%%%%%%%%%%%%%%%%%%%%%%%%%%%%%%%%%%%%%
%%%%%%%%%%%%%%%%%%%%%%%%%%%%%%%%%%%%%%%%%%%%%%%%%%%%%%%%%%%%%%%%%%%%%%

\paragraph{Algorithm analysis}\ \smallskip

\noindent Recall that $\varphi_i$, $v_i$, $R_i$, and
$V_i$ respectively denote the partial acyclic coloring obtained
after $i$ steps, the current vertex $v$ of the $i^{\rm th}$
step, the record $R$ after $i$ steps, and the input vector $V$
restricted to its $i$ first elements.

\medskip
We first show that the function defined by \ac is injective.

\begin{lemma}\label{lem:vo}
$V_i$ can be recovered from $(\varphi_i,R_i)$.
\end{lemma}

\begin{proof}
  First note that, at each step of Algorithm~\ac, a
  \algo{Color} line possibly followed by an \algo{Uncolor} line is
  appended to $R$. We will say that a step which only appends a
  \algo{Color} line is a \emph{color step}, and a step which appends a
  \algo{Color} line followed by an \algo{Uncolor} line is an
  \emph{uncolor step}. Therefore, by looking at the last line of $R$,
  we know whether the last step was a color step or an uncolor step.

  We first prove by induction on $i$ that $R_i$ uniquely determines the
  set of colored vertices at Step~$i$ (i.e. $\ovphi_i$). Observe that $R_1$ necessarily
  contains only one line which is \algo{Color}; then $v_1$ is the unique
  colored vertex. Assume now that $i\ge 2$. By induction hypothesis,
  $R_{i-1}$ (obtained from $R_i$ by removing the last line if Step $i$
  was a color step or by removing the two last lines if Step $i$ was
  an uncolor step) uniquely determines the set of colored vertices at
  Step $i-1$. Then at Step $i$, the smallest uncolored vertex of $G$
  is colored. If one of Events~1 to~4 happens, then the last line
  of $R_i$ is an \algo{Uncolor} line whose indicates which vertices
  are uncolored. Therefore, $R_i$ uniquely determines the set of
  colored vertices at Step $i$.
  
  Let us now prove by induction that the pair $(\varphi_i,R_i)$
  permits to recover~$V_i$.  At Step~1, $(\varphi_1,R_1)$ clearly
  permits to recover $V_1$: indeed, $v_1$ is the unique colored vertex
  and thus $V[1]=\varphi_1(v_1)$. Assume now that $i\ge 2$. The record
  $R_{i-1}$ gives us the set of colored vertices at Step $i-1$, and
  thus we know what is the smallest uncolored vertex $v$ at the
  beginning of Step $i$. Consider the following two cases:

  \begin{itemize}
  \item If Step $i$ was a color step, then $\varphi_{i-1}$ is obtained
    from $\varphi_i$ in such a way that
    $\varphi_{i-1}(u)=\varphi_i(u)$ for all $u\neq v$ and
    $\varphi_{i-1}(v)=\bullet$. By induction hypothesis,
    $(\varphi_{i-1},R_{i-1})$ permits to recover $V_{i-1}$ and
    $V[i]=\varphi_i(v)$.
    
  \item If Step $i$ was an uncolor step, then the last line of $R_i$
    allows us to determine the set of uncolored vertices at Step $i$
    and therefore, we can deduce $\varphi_{i-1}$. Then by induction
    hypothesis, $(\varphi_{i-1},R_{i-1})$ permits to recover
    $V_{i-1}$. We obtain $V[i]$ by considering the following cases:

    \begin{itemize}
    \item If the last line is of the form \algo{Uncolor, neighbor
        $u$}, then $V[i]=\varphi_i(u)$.
    \item If the last line is of the form \algo{Uncolor, special
      $u$}, then $V[i]=\varphi_i(u)$.
    \item If the last line is of the form \algo{Uncolor, cycle 
        $(u_1,u_2,u_3,u_4)$}, then $V[i]=\varphi_i(u_3)$.
    \item If the last line is of the form \algo{Uncolor, path
        $(u_1,u_2,u_3,u_4,u_5,u_6)$}, then $V[i]=\varphi_i(u_6)$.
    \end{itemize}
  \end{itemize}
  This completes the proof.
\end{proof}

\begin{lemma}\label{lem:r}
  Algorithm~\ac produces at most $o(\kappa^t)$ distinct records.
\end{lemma}

\begin{proof}
  As Algorithm~\ac fails to color $G$, the record $R$ has exactly $t$
  \algo{Color} steps. It contains also \algo{Uncolor} lines of
  different types: \algo{neighbor} (type $N$), \algo{special} (type
  $S$), \algo{cycle} (type $C$), and \algo{path} (type $P$). Let
  $\type=\set{N,S,C,P}$ be the set of bad event types.  Let denote
  $s_j$ the number of uncolored vertices when a bad event of type $j$
  occurs. Note that each
  \algo{Uncolor} step of type \algo{neighbor} (resp. \algo{special},
  \algo{cycle}, and \algo{path}) uncolors 1 (resp. 1, 2, 4) previously
  colored vertex. Hence set $s_N=1$, $s_S=1$, $s_C=2$ and $s_P=4$.
  
  To compute the total number of possible records, let us compute how
  many different entries, denoted $C_j$, an \algo{Uncolor} step of
  type $j$ can produce in the record. By considering vertex $v$ in
  \ac, observe that:
  \begin{itemize}
\item An \algo{Uncolor} step of type \algo{neighbor} can produce
  $\Delta$ different entries in the record, according to the
  neighbor of $v$ that shares the same color; hence let $C_N=\Delta$.
  
\item An \algo{Uncolor} step of type \algo{special} can produce
  $|S(v)| \le \alpha\Delta^\frac43$ different entries in the record,
  according to the vertex $u\in S(v)$ that shares the same color;
  hence let $C_S = \alpha\Delta^\frac43$.

\item An \algo{Uncolor} step of type \algo{cycle} can produce as many
  different entries in the record as the number of $4$-cycles going
  through $v$ and avoiding $S(v)$. We do not consider bicolored
  $4$-cycles going through $v$ and some vertex $u\in S(v)$, since we
  would have an \algo{Uncolor, special $u$} step instead. Hence this
  number of entries is bounded by $\frac{\Delta^\frac83}{8\alpha}$
  according to the next claim, and thus let $C_C = \frac{\Delta^\frac83}{8\alpha}$.

  \begin{claim}\label{lem:nb-C4}
    Given a graph $G$ with maximum degree $\Delta$, for any vertex $v$ of
    $G$, there are at most $\frac{\Delta^\frac83}{8\alpha}$ induced $4$-cycles
    going through $v$ and avoiding $S(v)$.
  \end{claim}

  \begin{proof}
    There are at most $\Delta^2$ edges between $N(v)$ and
    $N^2(v)$. Let $d$ be an integer such that $\deg(v,u)\ge d$
    if and only if $u\in S(v)$. Therefore, there are at least
    $d|S(v)|$ edges between $N(v)$ and $S(v)$.
    Thus there are at most $\Delta^2 - d\alpha\Delta^\frac43$ edges between
    $N(v)$ and $\overline S(v) = N^2(v)\setminus S(v)$, and
    \begin{eqnarray}
      \label{eqn:sum_deg}
      \sum_{u\in \overline S(v)} \deg(v,u) \le \Delta^2 - d\alpha\Delta^\frac43
    \end{eqnarray}
    One can see that the set of induced $4$-cycles passing through $v$
    and through some vertex $u\in N^2(v)$ is in bijection with the
    pairs of edges $\{ux,uy\}$ with $x\neq y$ and
    $\{x,y\}\subseteq N(v,u)$.  Thus there are
    ${\deg(v,u)}\choose{2}$ such cycles. Summing over all vertices in
    $\overline S(v)$, we can thus conclude that
    this is less than the
    following value $K = \frac{1}{2}\sum_{u\in \overline S(v)}
    \deg(v,u)^2$.  As this function is quadratic in $\deg(v,u)$, and as
    here $\deg(v,u) \le d$, Equation~(\ref{eqn:sum_deg}) implies that $K\le
    K(d)$ for $K(d) = \frac{1}{2}(\Delta^2 -
      d\alpha\Delta^\frac43)d$.  By simple calculation one can see that
    the polynomial $K(d)$ is maximal for $d = \frac{\Delta^\frac23}{2\alpha}$
    and we thus have that $K \le K\left(\frac{\Delta^\frac23}{2\alpha}\right) =
    \frac{\Delta^\frac83}{8\alpha}$.  This concludes the proof of the claim.
  \end{proof}
\item An \algo{Uncolor} step  of type \algo{path} can produce as many
  different entries in the record as the number of $6$-paths
  $P=(u_1,u_2,u_3,u_4,u_5,u_6)$ with $u_2=v$ and $u_1\prec u_3$.
  Hence this number of entries is bounded by
  $\frac{1}{2}\Delta(\Delta-1)^4$ according to the next claim, and
  thus let $C_P = \frac{1}{2}\Delta(\Delta-1)^4$.
  \begin{claim}\label{lem:nb-P6}
    Given a graph $G$ with maximum degree $\Delta$, for any vertex $v$ of $G$,
    there are at most $\frac{1}{2}\Delta(\Delta-1)^4$ paths
    $(u_1,u_2,u_3,u_4,u_5,u_6)$ of length
    6 with $u_2=v$ and $u_1\prec u_3$.
  \end{claim}
  \begin{proof}
    Given vertex $v$, there are ${\Delta \choose 2}$ possibilities to choose $u_1$
    and $u_3$, and then $\Delta-1$ candidates for being vertex
    $u_{i+1}$ once $u_i$ is known ($i\ge 3$). This clearly leads to the given
    upper bound.
  \end{proof}
\end{itemize}

We complete the proof by means of Theoremm~\ref{thm:nb_of_rec} of
Section~\ref{sec:general} (see \vpageref[below]{thm:nb_of_rec}).  Let
us consider the following polynomial $Q(x)$:
\begin{eqnarray*}
Q(x) &=& 1+ \sum_{i\in\type} C_i x^{s_i}\\
&=& 1+ C_N x^{s_N} + C_S x^{s_S} + C_C x^{s_C} + C_P x^{s_P}\\
&=& 1+ \Delta x + \alpha\Delta^{\frac43}x + \frac{\Delta^{\frac83}}{8\alpha}x^2+\frac12\Delta(\Delta-1)^4x^4\\
\end{eqnarray*}
Setting $X=\frac{2\sqrt{2\alpha}}{\Delta^\frac43}$, we have:
\begin{small}
\begin{equation}\label{eq:pierpauljack}\frac{Q(X)}{X} = 
\left ( \frac{1}{\sqrt{2\alpha}}+\alpha\right)\Delta^{\frac{4}{3}}+\left( 8
\alpha^{\frac{3}{2}} \sqrt{2} + 1 \right)\Delta - 32
\alpha^{\frac{3}{2}}\sqrt{2} +
\frac{8\alpha^{\frac{3}{2}}\sqrt{2}}{\Delta}\left( 6- \frac{4}{\Delta} + \frac{1}{\Delta^2}\right)
\end{equation}
\end{small}
In order to minimize $\frac{1}{\sqrt{2\alpha}}+\alpha$, we set
$\alpha=\frac12$, giving $X=\frac2{\Delta^{\frac43}}$
and we obtain:
\begin{equation*}
\frac{Q(X)}{X}=\frac{3}{2}\Delta^{\frac{4}{3}}+ 5\Delta - 16 + \frac{24}{\Delta}
- \frac{16}{\Delta^2} + \frac{4}{\Delta^3} <
\frac{3}{2}\Delta^{\frac{4}{3}}+ 5\Delta - 15  \le \kappa \mbox{ as soon as
  $\Delta\ge 24$}
\end{equation*}
Since $0<X\le 1$ for $\Delta\ge 24$, Algorithm~\ac produces at most
$o(\kappa^t)$ different records by Theorem~\ref{thm:nb_of_rec}. This completes the proof.
\end{proof}

\begin{remark}\label{re:eminem}
  For small values of $\Delta$, note that setting $\alpha=\frac12$ is not
  optimal. Indeed the best choice of $\alpha$ is the value minimizing the
  right term of Equation~(\ref{eq:pierpauljack}). For example, for
  $\Delta=27$, setting $\alpha = 0.225$ leads us to $194$ colors instead
  of $242$, already improving on Kostochka and Stocker's bound $ 1 +
  \left\lfloor\frac{(\Delta+1)^2}{4}\right\rfloor = 197$. Actually one
  can observe in Table~\ref{tab:cs} that the optimal value of $\alpha$
  (for a given $\Delta$) converges to $\frac12$ rather slowly.

  \begin{table}[h]
    \begin{center}
      \begin{tabular}{|c|c|c|c|c|c|c|c|c|c|c|}
        \hline
        $\Delta$ & 27 & 28 & 29 & 30 & 100 & 1000 & 10000 & 100000 & 1000000 \\
        \hline
        $\alpha$ & 0.225 & 0.225 & 0.226 & 0.226 & 0.25 & 0.32 &
        0.384 & 0.434 & 0.465  \\
        \hline
      \end{tabular}
      \caption{Optimal values of $\alpha$ for some given $\Delta$.\label{tab:cs}}
    \end{center}
  \end{table}
  
\end{remark} 

\subsubsection{A better upper bound for large value of $\Delta$}\label{subsubsec:algobis}

%%%%%%%%%%%%%%%%%%%%%%%%%%%%%%%%%%%%%%%%%%%%%%%%%%%%%%%%%%%%%%%%%%%%%%
%%%%%%%%%%%%%%%%%%%%%%%%%%%%%%%%%%%%%%%%%%%%%%%%%%%%%%%%%%%%%%%%%%%%%%
%%                           ALGORITHM  (Cycle)                     %%
%%%%%%%%%%%%%%%%%%%%%%%%%%%%%%%%%%%%%%%%%%%%%%%%%%%%%%%%%%%%%%%%%%%%%%
%%%%%%%%%%%%%%%%%%%%%%%%%%%%%%%%%%%%%%%%%%%%%%%%%%%%%%%%%%%%%%%%%%%%%%

The choice of the bad event types is important and considering two different sets
of bad event types (insuring the acyclic coloring property) may lead to
different bounds. In the previous subsection, we have considered four bad
event types that insure a coloring to be acyclic. In this subsection,
we consider an other set of bad event types which leads to a better
upper bound for large value of $\Delta$.

Algorithm~\acb~(see \vpageref[above]{algo:ac2}) is a variant of
 Algorithm~\ac~(see \vpageref[below]{algo:ac}) based on the following
 set of three bad events:
\begin{enumerate}[Event 1:]
\item[Event $N$:] $G$ contains a monochromatic edge $vu$ for some $u$ (line 8 of
  the algorithm);
\item[Event $S$:] $G$ contains a special couple $(v,u)$ with
  $u$ and $v$ having the same color (line 11 of the algorithm);
\item[Event $k$:] $G$ contains a bicolored cycle of length $2k$ $(u_1,u_2=v,u_3,\ldots,u_{2k})$
  (line 14 of the algorithm);
\end{enumerate}

\begin{algorithm}[t]
  \DontPrintSemicolon
  \LinesNumbered
  
  \SetKwInOut{Input}{Input}
  \SetKwInOut{Output}{Output}
  \Input{$V$ (vector of length $t$).}
  \Output{($\varphi$, $R$).\vspace{.2cm}}
  
  \For{all $v$ in $V(G)$}{
    $\varphi(v)\gets \bullet$ 
  }
  $R\gets newfile()$

  \For{$i\gets 1$ \KwTo $t$}{

    Let $v$ be the smallest (w.r.t. $\prec$) uncolored vertex of $G$
    
    $\varphi(v) \gets V[i]$
    
    Write "\texttt{Color} $\backslash$n" in $R$
    
    \If{$\varphi(v)=\varphi(u)$ \rm{for} $u\in N(v)$} %% proper
    {
      \tcp{Proper coloring issue}
      $\varphi(v) \gets \bullet$
      
      Write "\texttt{Uncolor, neighbor $u$ $\backslash$n}" in $R$
      
    }
    \ElseIf{$\varphi(v)=\varphi(u)$ \rm{for} $u\in S(v)$}  %% special couple
    {
      \tcp{Special couple issue}
      $\varphi(v) \gets \bullet$
      
      Write "\texttt{Uncolor, special $u$} $\backslash$n" in $R$
    }
    \ElseIf{$v$ belongs to a bicolored cycle of length $2k$ ($k\ge
      2$), say $(u_1,u_2=v,u_3,\ldots,u_{2k})$ with $u_1\prec u_3$} %% C_4
    {
      \tcp{Bicolored cycle issue}

      \For{$j\gets 1$ \KwTo $2k-2$}{
        $\varphi(u_j) \gets \bullet$
      }
      Write "\texttt{Uncolor, cycle $(u_1,\ldots,u_{2k})\ \backslash$n}" in $R$
    }
  }
  \Return ($\varphi$, $R$)
  \caption{\acb\label{algo:ac2}}
\end{algorithm} 

\noindent This leads to the following upper
bound when $\Delta\ge 9$: $$\chi_a(G) < \nbmaxcoldeltabis.$$
Let
$\kappa$ be the unique integer such that $\frac32\Delta^\frac43 +
\Delta +\frac{8\Delta^\frac43}{\Delta^\frac23-4} \le \kappa<
\nbmaxcoldeltabis$ and let $\alpha=\frac12$. We now briefly sketch the
proof.  Let $\type=\set{N,S,2,3,4,\ldots,\frac n2}$ be the set of bad
event types. Note that each \algo{Uncolor} step of type
\algo{neighbor} (resp. \algo{special} and \algo{$2k$-cycle})) uncolors
$1$ (resp. $1$, $2k-2$) previously colored vertex. Hence set $s_N=1$,
$s_S=1$ and $s_k=2k-2$.

By considering $v$ in Algorithm \acb, observe that:

\begin{itemize}
\item An \algo{Uncolor} step of type \algo{neighbor} can produce
  $\Delta$ different entries in the record. Set $C_N=\Delta$.

\item An \algo{Uncolor} step of type \algo{special} can produce
  $|S(v)| \le \frac12\Delta^\frac43$ different entries in the record,
  according to the vertex $u\in S(v)$ that shares the same color. Set
  $C_S=\frac12 \Delta^\frac43$.

\item Now consider cycles of length $2k$, $k\ge 2$. For cycles of
  length 4, there are at most $\frac14\Delta^\frac83$ induced 4-cycles
  going through $v$ and avoiding $S(v)$ (see Claim \ref{lem:nb-C4});
  we set $C_2=\frac14\Delta^\frac83$. 

  Let $k\ge 3$. Let us upper bound the number of $2k$-cycles going
  through $v$ that may be bicolored. To do so, we count the number of
  $2k$-cycles $(u_1,u_2,u_3,\ldots,u_{2k})$ with $u_2=v$, $u_1 \prec
  u_3$ such that $(u_1,u_{2k-1})$ or $(u_{2k-1},u_1)$ is not special
  (if both $(u_1,u_{2k-1})$ and $(u_{2k-1},u_1)$ are special, then
  $u_1$ and $u_{2k-1}$ cannot receive the same color). There are at
  most $\Delta^{2k-\frac43}$ such cycles according to
  Claim~\ref{lem:riberi}. We set $C_{k}=\Delta^{2k-\frac43}$.

\begin{claim}\label{lem:riberi}
For $k\ge 3$, there are at most $\Delta^{2k-\frac43}$ $2k$-cycles
$(u_1,u_2,u_3,\ldots,u_{2k})$ going through $v$ with $v=u_2$ and
$u_1\prec u_3$ such that $(u_1,u_{2k-1})$ or $(u_{2k-1},u_1)$ is not
special.
\end{claim}

\begin{proof}
As $u_1\prec u_3$, given $v$, there are ${\Delta \choose 2}$ possible
$(u_1,u_3)$. Then knowing $u_i$, there are at most $\Delta$ possible
choices for $u_{i+1}$, $3\le i\le 2k-2$.  Now let $(r,s)$ be a
non-special pair being either $(u_1,u_{2k-1})$ or
$(u_{2k-1},u_1)$. Hence $s\in N^2(r)\backslash S(r)$. Let $d$ be the
highest value of $\deg(r,u)$ for $u\in N^2(r)\backslash
S(r)$. Therefore, there are at least $d|S(r)|$ edges between $N(r)$
and $S(r)$, and so at most $\Delta^2-\frac{d}{2} \Delta^\frac43$ edges
between $N(r)$ and $N^2(r)\setminus S(r)$. It follows that $d$ is at
most $2\Delta^\frac23$.  Hence, there are at most $2\Delta^\frac23$
possible choices for $u_{2k}$. This leads to the given upper bound.
\end{proof}
\end{itemize}
Let us consider the following polynomial $Q(x)$:
\begin{eqnarray*}
Q(x) &=& 1+ \sum_{i\in\type} C_i x^{s_i}\\
&=& 1 + C_N x^{s_N}+ C_S x^{s_S} + C_2x^{s_2}+ \sum_{k\ge
  3}^{\lfloor n/2\rfloor} C_{k}x^{s_k}\\
&=&1 + \Delta x + \frac12 \Delta^\frac43 x + \frac14\Delta^\frac83 x^2
+ \sum_{k\ge 3}^{\lfloor n/2\rfloor} \Delta^{2k-\frac43}x^{2k-2}\\
&<&1 + \Delta x + \frac12 \Delta^\frac43 x + \frac14\Delta^\frac83 x^2
+ \frac{\Delta^\frac{14}{3}x^4}{1-\Delta^2x^2} \qquad\qquad \;
\mbox{for} \; x<\frac1\Delta
\end{eqnarray*}

Setting $X=\frac{2}{\Delta^\frac43}$, we have $X\le \frac1\Delta$ as
soon as $\Delta \ge 9$ and thus: 
\begin{equation*}\frac{Q(X)}{X} < 
\frac32\Delta^\frac43+\Delta+\frac{8\Delta^\frac43}{\Delta^\frac23-4}
\le \kappa
\end{equation*}
Algorithm~\acb produces at most
$o(\kappa^t)$ different records by Theorem~\ref{thm:nb_of_rec}. This completes the sketch of the proof.

% Local variables:
% mode: latex
% TeX-master: "0_main.tex"
% End:
 %% section 2.2
\section{General method}\label{sec:general_la_vraie}

In the previous section, we gave upper bounds on the acyclic chromatic
number of some graph classes. To do so, we precisely analyzed the
randomized procedure for a specific graph class and a specific graph
coloring. The aim of this section is to provide a general method that
can be applied to several graph classes and many graph colorings (some
applications of our general method are given in Section~\ref{sec:application}).

\bigskip 

In the remaining of this section, $G$ is an arbitrarily chosen graph.
The aim of the general method is to prove the existence of a
particular coloring of $G$ using $\kappa$ colors, for some $\kappa$.
A {\em partial coloring} of $G$ is a mapping $\varphi : V(G) \to
\set{\bullet,1,2,\ldots,\kappa}$ ($\bullet$ means that the vertex is
uncolored). We assume by contradiction that $G$ does not admit such a
coloring. In that case, we will show that Algorithm~\gc (see
Algorithm~\ref{algo:gen}) defines an injective mapping
(Corollary~\ref{cor:injective}) from $\kappa^t$ different inputs (for
some $t$) to $o(\kappa^t)$ different outputs
(Theorem~\ref{thm:nb_of_rec}), a contradiction. Given a partial
coloring $\varphi$, let $\ovphi$ denotes the set of vertices colored
in $\varphi$.

\subsection{Description of Algorithm \gc}\label{subsec:requ-fun}

Given a vertex $v$ of $G$, let $\mathbb{F}(v)$ denote the set of {\em
  forbidden partial colorings anchored at $v$}. This set is such that
the vertex $v$ is colored for any $\varphi\in \mathbb{F}(v)$. For
example, Algorithm~\acm (see Algorithm~\ref{algo:acm}) is a special
case of Algorithm \gc, where, for any vertex $v$, the set
$\mathbb{F}(v)$ consists of the partial colorings where $v$ and one of
its neighbor have the same color, or $v$ belongs to a properly
bicolored cycle.

A partial coloring $\varphi$ of $G$ is said to be {\em allowed}, if
and only if,
\begin{enumerate}
\item either $\varphi$ is empty (none of the vertices is colored),
\item or there exists a colored vertex $v$ such that $\varphi \notin
  \mathbb{F}(v)$ and uncoloring $v$ yields to an allowed coloring.
\end{enumerate}

\begin{algorithm}[t]
  \DontPrintSemicolon \LinesNumbered
  
  \SetKwInOut{Input}{Input}
  \SetKwInOut{Output}{Output}
  \Input{$V=\set{1,2,\ldots,\kappa}^t$ (vector of length $t$).}
  \Output{($\varphi$, $R$).\vspace{.2cm}}
  
  \For{ all $v$ in $V(G)$}{
    $\varphi(v)\gets \bullet$ 
  }

  $R\gets {\rm newfile}()$

  \For{$i\gets 1$ \KwTo $t$}{

    $v\gets$ $\NUV(\ovphi)$
    
    $\varphi(v) \gets V[i]$
    
    Write "\texttt{Color $\backslash$n}" in $R$
    
    \If{$\varphi\in\mathbb{F}(v)$}
    {
      $j \gets $ $\BET(v,\varphi)$

      $k \gets $ $\ONBE_j(v,\varphi)$

      \For{$\forall u \in \USBE_j(v,\ovphi,k)$}{
        $\varphi(u) \gets \bullet$
      }
      Write "\texttt{Uncolor, Bad Event j, $k\ \backslash$n}" in $R$
    }
  }
  \Return ($\varphi$, $R$)
  \caption{\gc\label{algo:gen}}
\end{algorithm} 

Algorithm~\gc constructs a partial coloring $\varphi$ of $G$.  A
crucial invariant of Algorithm~\gc is that the partial coloring
$\varphi$ considered at the beginning of each iteration  of the main loop is allowed. 

At the beginning of each iteration, Algorithm~\gc starts with an allowed coloring $\varphi$ and chooses 
an uncolored vertex $v$ by $\NUV$.
\begin{itemize}
\item $\NUV(\ovphi)$: This function takes the set of colored vertices
  of $G$ in $\varphi$ as input and outputs an uncolored vertex (unless
  all vertices are colored).
\end{itemize}
Then Algorithm~\gc colors $v$ using the next color from vector $V$. This new coloring $\varphi$ either
verifies $\varphi\notin\mathbb{F}(v)$ and consequently $\varphi$ is
allowed, or $\varphi \in \mathbb{F}(v)$ and in that case $\varphi$ is
an ``almost'' allowed coloring since uncoloring $v$ yields an allowed
coloring.  Hence, let us define these forbidden colorings that can be
produced by Algorithm~\gc.

A partial coloring $\varphi$ of $G$ is said to be a {\em bad event
  anchored at $v$}, if $\varphi\in\mathbb{F}(v)$ and if the partial
coloring $\varphi'$, obtained from $\varphi$ by uncoloring $v$, is
such that
\begin{itemize}
\item $\varphi'$ is an allowed coloring,
\item $v$ is the vertex output by $\NUV(\ovphi')$.
\end{itemize}
We denote $\mathbb{B}(v)$ the set of bad events anchored at $v$.  It
is clear that $\mathbb{B}(v)\subseteq\mathbb{F}(v)$. Hence, the colorings $\varphi$
considered at line~8 of the algorithm are either allowed or belong to
$\mathbb{B}(v)$. Therefore, the test at line~8 is thus equivalent to testing whether
$\varphi\in\mathbb{B}(v)$. 

Before going further into the description of \gc, let us introduce the
following refinements of the sets $\mathbb{B}(v)$.  For some set
$\type$, each set $\mathbb{B}(v)$ is partitioned into $|\type|$ sets
$\mathbb{B}_j(v)$ where $j \in \type$.  We call the bad events of
$\mathbb{B}_j(v)$ the {\em type $j$ bad events}.  We now refine again
each set $\mathbb{B}_j(v)$.  We partition each $\mathbb{B}_j(v)$ into
different classes $\mathbb{B}_j^k(v)$ where $k$ belongs to some set
$\class_j(v)$ of cardinality at most $C_j$, for some value $C_j$
(depending only on type $j$). The partition into classes must be
sufficiently refined in order to allow some properties of the function
$\RBE$ (see below). 

After coloring $v$ in the main loop, if the current coloring
$\varphi$ does not belong to $\mathbb{B}(v)$, then \gc proceeds to the next
iteration. Observe that in that case $\varphi$ remains allowed as
expected.

Suppose now that after coloring $v$, the current coloring $\varphi$
belongs to $\mathbb{B}(v)$.  In that case, \gc determines the values
$j$ and $k$ such that $\varphi\in\mathbb{B}_j^k(v)$. That is done
using the following two functions:

\begin{itemize}
\item $\BET(v,\varphi)$: When
  $\varphi$ is a bad event of $\mathbb{B}(v)$, this function
  outputs the element $j\in \type$ such that  $\varphi$ is
  a bad event belonging to $\mathbb{B}_j(v)$. 
\end{itemize}

\begin{itemize}
\item $\ONBE_j(v,\varphi)$ for some $j \in \type$: When
  $\varphi$ is a bad event of $\mathbb{B}_j(v)$, this function
  outputs the element $k\in \class_j(v)$ such that  $\varphi$ is
  a bad event belonging to $\mathbb{B}_j^k(v)$. 
\end{itemize}

Then \gc uncolors the vertices given by $\USBE$, and
proceeds to the next iteration. A key property of $\USBE$ is to ensure that
the obtained coloring (i.e. obtained after uncoloring the vertices
given by $\USBE$) is allowed as expected. 

\begin{itemize}
\item $\USBE_j(v,\ovphi,k)$ for some $j \in \type$: For any bad event
  $\varphi$ of $\mathbb{B}_j^k(v)$ (with colored vertices $\ovphi$),
  this function outputs a subset $S$ of $\ovphi$ of size $s_j$ (for
  some value $s_j$ depending only on type $j$), such that uncoloring
  the vertices of $S$ in $\varphi$ yields an allowed coloring.
\end{itemize}
Often the property of leading to an allowed coloring is easy to
fulfill (see Lemma~\ref{lem-USBE}). A set $X$ of partial colorings of
$G$ is \emph{closed upward} (resp. \emph{closed downward}) if starting
from any partial coloring of $X$, coloring (resp. uncoloring) any
uncolored (resp. colored) vertex leads to another coloring of $X$.
\begin{lemma}\label{lem-USBE}
  If every set $\mathbb{F}(u)$ is closed upward, then the set of allowed
  colorings is closed downward. Hence in that case, for any $\varphi\in
  \mathbb{B}(v)$, uncoloring a set $S$ of vertices containing $v$, leads
  to an allowed coloring.
\end{lemma}
\begin{proof}
  Let us first prove the first statement. Assume for contradiction
  that the set of allowed colorings is not closed downward, that is  there exist an allowed coloring $\varphi$ and a non-empty set
  $S\subset \ovphi$, such that uncoloring the vertices in $S$ leads to
  a non-allowed coloring $\varphi'$.  As $\varphi$ is allowed, there
  exists an ordering $v_1,\ldots,v_p$, with $p=|\ovphi|$, of the
  vertices in $\ovphi$ such that the restriction of $\varphi$ to
  vertices $v_1,\ldots,v_i$, denoted $\varphi_i$, does not belong to
  $\mathbb{F}(v_i)$, for any $i\le p$. Let us denote $\varphi'_i$ the
  coloring obtained from $\varphi_i$ by uncoloring the vertices of $S$
  (if colored). As $\varphi'$ is not allowed, there exists a value
  $1\le j\le p$ such that $\varphi'_j\in \mathbb{F}(v_j)$.  But as
  $\mathbb{F}(v_j)$ is closed upwards, this contradicts the fact that
  $\varphi_j\notin \mathbb{F}(v_j)$.

  Consider now the second statement. For any $\varphi\in
  \mathbb{B}(v)$, uncoloring $v$ leads to an allowed coloring (by
  definition of $\mathbb{B}(v)$). Then the proof follows from the fact
  that allowed colorings are closed downward.
\end{proof}

Finally, to prove the injectivity of \gc, we need that the following
function exists.
\begin{itemize}
\item $\RBE_j(v,X,k,\varphi')$ where $X \subseteq V(G)$, $k\in
  \class_j(v)$, and $\varphi'$ is a partial coloring of $G$: The
  function outputs a bad event $\varphi \in \mathbb{B}_j^k(v)$, such
  that (1) $\ovphi = X$ and (2) uncoloring $\USBE_j(v,\ovphi,k)$ from
  $\varphi$ one obtains $\varphi'$, if such partial coloring $\varphi$
  exists.  Moreover, the partition into classes of $\mathbb{B}_j(v)$
  must be sufficiently refined so that at most one bad event $\varphi$
  fulfills these conditions.
\end{itemize}

\paragraph{Example}\ \smallskip

Let us illustrate our general method with the proofs of
Section~\ref{sec:acyclic} on acyclic vertex-coloring.

Observe that Algorithm~\ref{algo:acm} corresponds to
Algorithm~\ref{algo:gen} for the following settings. For any vertex
$v$, the set $\mathbb{F}(v)$ contains every partial coloring of $G$
with a monochromatic edge or with a bicolored cycle involving $v$.
Then one type (type $1$) corresponds to monochromatic edges, and
several types (type $k$, for $k\ge 2$) correspond to bicolored cycles,
one per possible length of the cycles. Then each type is partitionned
into classes, each of them corresponding to one monochromatic edge or
to one bicolored cycle, respectively. For the uncoloring process, one
can notice that the number of uncolored vertices only depends on the
type of bad events, $s_1=1$ and $s_{k}=2k-2$, and that the set of
uncolored vertices only depend on the class (i.e. the monochromatic
edge or the bicolored cycle). Furthermore, as the sets $\mathbb{F}(v)$
are closed upward and as the current vertex is always uncolored, at
the end of each iteration the partial colorings are always allowed (by
Lemma~\ref{lem-USBE}). Finally, as described in
Subsection~\ref{ssec:acyclic-gamma} there exists a function $\RBE_j$
for each type of bad event $j$.

Similarly, Algorithm~\ref{algo:ac} also corresponds to
Algorithm~\ref{algo:gen}. Here, $\mathbb{F}(v)$ contains every
partial coloring of $G$ with a monochromatic edge $vu$, a
monochromatic special pair $(v,u)$, a properly bicolored $4$-cycle
$(v,u_1,u_2,u_3)$ or a properly bicolored $6$-path
$(u_1,v,u_3,u_4,u_5, u_6)$ with $u_1\prec u_3$.

\subsection{Algorithm \gc and its analysis}

From the previous subsection, we have that for $j\in \type$, $C_j$
and $s_j$ respectively denote the number of type $j$ bad event
classes, and the number of vertices to be uncolored when a type
$j$ bad event occurs. We set
$$Q(x) = 1+ \sum_{j\in\type} C_j x^{s_j}$$

In this subsection, we prove the following:
\begin{theorem}\label{th:main-final}
  The graph $G$ admits an allowed $\kappa$-coloring for any integer $\kappa$ such that 
  $$\kappa\ge
  \min_{0<x\le 1} \frac{Q(x)}x.$$ 
\end{theorem}

Before going further to prove Theorem~\ref{th:main-final}, let us state the two following remarks.

\begin{remark}
  One can observe that the bound obtained when
  all $s_j=1$, namely \linebreak $\kappa\ge 1+\sum_{j\in\type}C_j$, is the same as the
  one obtained by a simple greedy coloring. Indeed, while coloring the current vertex
  $v$, the bad events of type $j$ ``forbid'' at most $C_j$ colors for
  $v$, and so $1+\sum_{j\in\type}C_j$ colors suffice to color the
  graph greedily.
\end{remark}

\begin{remark}
  One can observe that the polynomial $Q(x)$ only depends on the
  values $\displaystyle X_k = \sum_{j\ {\rm s.t.}\ s_j=k}C_j$. One could thus merge
  the bad event types having the same value $s_j$.
\end{remark}

From now on, we
assume that $G$ does not admit an allowed $\kappa$-coloring, this will
lead to a contradiction.  Let $V\in\set{1,2,\ldots,\kappa}^t$ be a
vector of length $t$ for some arbitrarily large $t$.  The
algorithm~\gc (see Algorithm~\ref{algo:gen}) takes the vector $V$ as
input and returns an allowed partial coloring $\varphi$ of $G$ and a
text file $R$ (called the \emph{record}).  Let $\varphi_i$, $v_i$,
$R_i$, and $V_i$ respectively denote the partial coloring obtained by
Algorithm~\gc after $i$ steps, the current vertex $v$ of the
$i^{\rm th}$ step, the record $R$ after $i$ steps, and the input
vector $V$ restricted to its $i$ first elements.  Note that the
algorithm and especially the properties of $\USBE_j(v,\ovphi,k)$
ensure that each $\varphi_i$ is allowed.  As $\varphi_i$ is an allowed
partial $\kappa$-coloring of $G$ and since $G$ has no allowed
$\kappa$-coloring by hypothesis, we have that $\ovphi_i \subsetneq
V(G)$ and that vertex $v_{i+1}$ is well defined. This also implies that $R$
has $t$ \algo{Color} lines. Finally note that $R_i$ corresponds to the
lines of $R$ before the $(i+1)^{\rm th}$ \algo{Color} line.

\begin{lemma}\label{lem:grosMinet}
  One can recover $v_i$ and $\ovphi_i$ from $R_i$.
\end{lemma}

\begin{proof}
  By induction on $i$.  Trivially, $\ovphi_0=\emptyset$ and $v_0$ does
  not exist. Consider now $R_{i+1}$ and let us show that we can
  recover $v_{i+1}$ and $\ovphi_{i+1}$. To recover $R_{i}$ from
  $R_{i+1}$ it is sufficient to consider the lines before the last
  (i.e. the $(i+1)^{\rm th}$) \algo{Color} line.
  By induction hypothesis, one can recover $\ovphi_i$ from
  $R_i$. Observe that $v_{i+1}=\NUV(\ovphi_i)$. Let $X=\ovphi_i +
  v_{i+1}$. If the last line of $R_{i+1}$ is a \algo{Color} line, then
  $\ovphi_{i+1}= X$. Otherwise, the last line of $R_{i+1}$ is
  an \algo{Uncolor} line of the form \algo{Uncolor, Bad Event $j$,
    $k$}. Then, we have $\ovphi_{i+1} =
  X\setminus \USBE_j(v_{i+1},X,k)$. That completes the proof.
\end{proof}

\begin{lemma}\label{lem:vo2bis}
  One can recover $V_i$ from $(\varphi_i,R_i)$.
\end{lemma}

\begin{proof} By induction on $i$. Trivially, $V_0$ (which is empty) 
  can be recovered from $(\varphi_0,R_0)$.  Consider now
  $(\varphi_{i+1},R_{i+1})$ and let us try to recover $V_{i+1}$. By
  induction, it is thus sufficient to recover $R_{i}$, $\varphi_{i}$,
  and the value $V[i+1]$.  As previously seen in the proof of
  Lemma~\ref{lem:grosMinet}, we can deduce $R_i$ from $R_{i+1}$.  By
  Lemma \ref{lem:grosMinet}, we know $\ovphi_{i}$ and we have
  $v_{i+1}=\NUV(\ovphi_i)$.  Note that in the $(i+1)^{\rm th}$ step of
  Algorithm~\gc, we wrote one or two lines in the record:
  exactly one \algo{Color} line followed either by nothing, or by one
  \algo{Uncolor, Bad Event $j$, $k$} line. Let us consider these two
  cases separately:

  \begin{itemize}  
  \item If Step $i+1$ was a color step alone, then $V[i+1]
    = \varphi_{i+1}(v_{i+1})$ and $\varphi_{i}$ is obtained from $\varphi_{i+1}$
    by uncoloring $v_{i+1}$.
    
  \item If the last line of $R_{i+1}$ is \algo{Uncolor, Bad Event $j$,
      $k$}, then the function \linebreak
    $\RBE_j(v_{i+1},\ovphi_i,k,\varphi_{i+1})$ outputs the bad event
    $\varphi'_i$ that occured during this step of the algorithm. Then we
    have that $V[i+1]=\varphi'_i(v_{i+1})$ and that $\varphi_i$
    corresponds to the partial coloring obtained from $\varphi'_i$ by
    uncoloring $v_{i+1}$.
  \end{itemize} 
  This concludes the proof of the lemma.
\end{proof}

\begin{corollary}\label{cor:injective}
  The mapping $V \rightarrow (\varphi,R)$ defined by Algorithm \gc   is injective.
\end{corollary}

\begin{proofof}{Theorem~\ref{th:main-final}}
  First observe that Algorithm \gc can produce at most
  $\smallO{\kappa^t}$ distinct outputs $(\varphi,R)$; indeed, there are at
  most $(1+\kappa)^n$ partial colorings $\varphi$ of $G$ and at most
  $\smallO{\kappa^t}$ records $R$ (by Theorem~\ref{thm:nb_of_rec}, see
  Section~\ref{sec:general}).  This is
  less than the $\kappa^t$ possible inputs (for a sufficiently large
  $t$), and thus contradicts the injectivity of Algorithm \gc (Corollary
  \ref{cor:injective}).  This concludes the proof.
\end{proofof}

\subsection{Extension to list-coloring}

Given a graph $G$ and a list assignment $L(v)$ of colors for every
vertex $v$ of $G$, we say that $G$ admits a $L$-coloring if there is a
vertex-coloring such that every vertex $v$ receives its color from its
own list $L(v)$. A graph is \emph{$k$-choosable} if it is
$L$-colorable for any list assignment $L$ such that $|L(v)|\ge k$ for
every $v$. The minimum integer $k$ such that $G$ is $k$-choosable is
called the \emph{choice number} of $G$. The usual coloring is a
particular case of $L$-coloring (all the lists are equal) and thus the
choice number upper bounds the chromatic number. This notion naturally
extends to edge-coloring and chromatic index.

\bigskip

Until now, our methods were developed for usual colorings
(i.e. without lists). Every algorithm takes a vector of colors $V$ as
input and, at each Step $i$, a vertex $v$ is colored with color $V[i]$
(line 6 of Algorithm~\gc).  It is easy to slightly modify our
procedure to extend all our results to list-coloring. To do so, the
input vector $V$ is no longer a vector of colors but a vector of
indices. Then, at each Step $i$, the current vertex $v$ is
colored with the $V[i]^{\rm{th}}$ color of $L(v)$. We then adapt the
proof of Lemma~\ref{lem:vo2bis} so that $V[i+1]$ is no longer
$\varphi_{i+1}(v_{i+1})$ (or $\varphi'_{i}(v_{i+1})$) but instead it
is the position of $\varphi_{i+1}(v_{i+1})$ (or
$\varphi'_{i}(v_{i+1})$) in $L(v_{i+1})$.

\bigskip

Therefore, Theorems~\ref{thm:acy-delta},~\ref{thm:gamma},
and~\ref{th:main-final} extend to list-coloring. 

% Local variables:
% mode: latex
% TeX-master: "0_main.tex"
% ispell-local-dictionary: "english"
% End:
%% section 3
\section{Bounding the number of records}
\label{sec:general}

The aim of this section is to prove one of our main theorems, namely
Theorem~\ref{thm:nb_of_rec}, that upper bounds the number of possible
records produced by Algorithm~\gc.

\bigskip

Let us define a class of records $\mathcal{R}$ which includes the
records that Algorithm~\gc could produce in a real execution.  In this
section, let $n=|V(G)|$ be the order of the graph $G$, $\type$ be a
set of bad event types, and $s_j$ and $C_j$ be positive integers for
all $j\in \type$, corresponding to the number of uncolored vertices
and the number of classes associated to the bad events of type $j$.

A record $R\in\mathcal{R}$ is a sequence of \algo{Color} and
\algo{Uncolor, Bad Event $j$, $k$} lines, where $j\in \type$ and
$k\in \{1,\ldots,C_j\}$. The \emph{Dyck paths} are defined as
staircase lattice paths on a square grid, from the lower-left corner
to the upper-right corner, which do not go below the diagonal. We say
that a Dyck path is \emph{partial} when it does not end in the
upper-right corner. The \emph{size} of a (partial) Dyck path is its
number of up-steps. Observe that a record $R\in\mathcal{R}$
can be seen as a \emph{partial Dyck path} where
\begin{itemize}
\item each up-step corresponds to a \algo{Color} line,
\item each descent (maximal sequence of consecutive down-steps) of
  length $\ell$ is annotated with a couple $(j,k)$ and corresponds to
  an \algo{Uncolor, Bad Event $j$, $k$} line where $\ell = s_j$.
\end{itemize}

\begin{figure}[t]
  \subfigure[]{\begin{minipage}{3.5cm}\fontsize{7}{7}\selectfont\begin{enumerate}[]
        \setlength{\itemsep}{-0.1cm}
      \item Color
      \item Color
      \item Uncolor, Bad Event $j_1$, $k_1$ 
      \item Color 
      \item Uncolor, Bad Event $j_2$, $k_2$
      \item Color 
      \item Color 
      \item Color 
      \item Color 
      \item Uncolor, Bad Event $j_3$, $k_3$
      \item Color 
      \item Color 
      \item Uncolor, Bad Event $j_4$, $k_4$
      \item Color 
      \end{enumerate}\end{minipage}} \hfill
  \subfigure[]{\begin{minipage}{10cm}\includegraphics[scale=0.45]{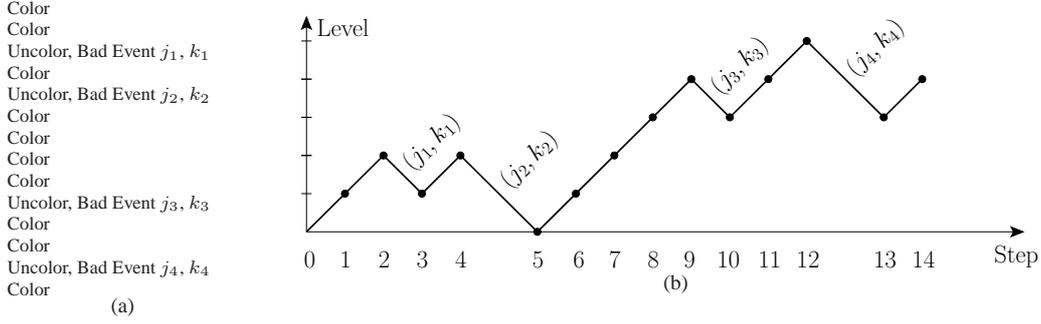}\end{minipage}}
  \caption{(a) A record and (b) its corresponding annotated partial Dyck
    path.}
  \label{fig:Dyck}
\end{figure} 

Observe Figure~\ref{fig:Dyck} which gives an example of such an
annotated partial Dyck path where $s_{j_1} = 1$, $s_{j_2} = 2$,
$s_{j_3} = 1$, and $s_{j_4} = 2$.

From now on, the term \emph{record} refers to both a record produced by
Algorithm~\gc and its corresponding annotated partial Dyck path.

At a given step, it is clear that the level of the record corresponds
to the number of colored vertices in $G$ (for example, at Step $8$ of
Figure~\ref{fig:Dyck}, the graph $G$ has $3$ colored vertices). Thus
the ending level of the record should be between $0$ and $n$. Let us
define the subclass $\mathcal{B} \subseteq \mathcal{R}$ of the records
ending at level 0. In the following, usual Dyck paths will be called
\emph{non-partial} Dyck path to emphasize the difference between Dyck
paths and partial Dyck paths. Hence, $\mathcal{B}$ is the set of
non-partial Dyck paths of $\mathcal{R}$.

It is clear that the size of a record of $\mathcal{R}$ is the
number of \algo{Color} lines. Let $r_t$
(resp. $b_t$) be the number of records of size $t$ in $\mathcal{R}$
(resp. $\mathcal{B}$) for any $t\ge 0$. We thus define the generating
functions of $\mathcal{R}$ and $\mathcal{B}$ as
$$R(y) = \sum_{t\ge 0} r_t y^t \quad \mbox{and} \quad B(y) =
\sum_{t\ge 0} b_t y^t.$$

Let $\mathcal{R}_\ell\subseteq \mathcal{R}$ be the set of records of
$\mathcal{R}$ ending at level $\ell$. Since during the execution
of Algorithm~\gc, every \algo{Uncolor} line follows a \algo{Color}
line, a record $R\in\mathcal{R}_\ell$
can be split into $\ell$ up-steps (which
correspond to the last up-steps between level $i$ and $i+1$, for each $0\le i\le\ell-1$) and
$\ell+1$ records
$\set{B_1,B_2,\ldots,B_{\ell+1}}\subseteq\mathcal{B}$ (See
Figure~\ref{fig:Dyck_R+B}). Hence, the generating function of
$\mathcal{R}_\ell$ is $R_\ell(y)=y^\ell B(y)^{\ell+1}$.
Therefore,

\begin{equation}\label{eqn:R+B}
  R(y) = \sum_{0\le \ell \le n} R_\ell(y) = \sum_{0\le \ell \le n} y^\ell B(y)^{\ell+1}
\end{equation}
\begin{figure}[htb]
  \begin{center}
    \includegraphics[scale=0.6]{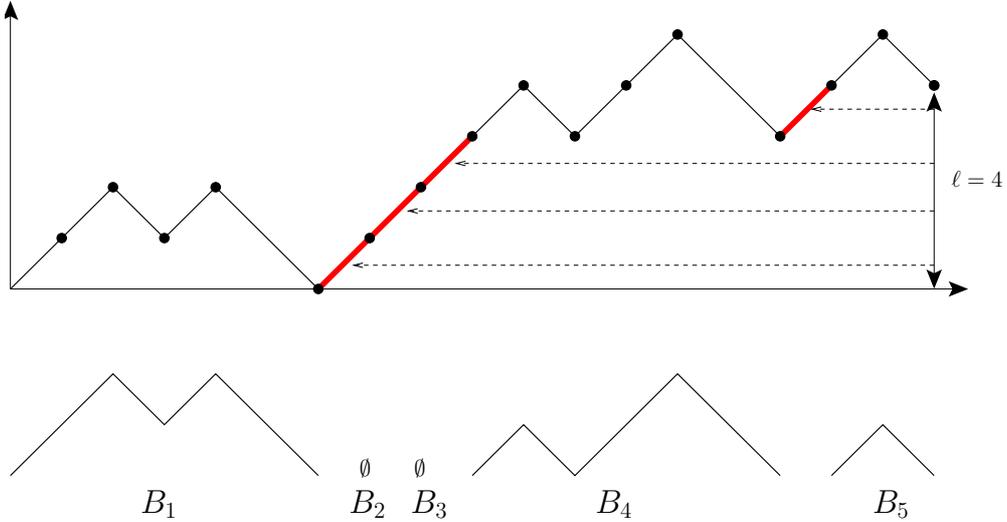}
    \caption{Splitting a partial Dyck path of level $\ell$ into $\ell+1$ non-partial Dyck paths and $\ell$ up-steps.}
    \label{fig:Dyck_R+B}
  \end{center}
\end{figure}

Let $\mathcal{B}_j\subseteq \mathcal{B}$ be the set of records of
$\mathcal{B}$ ending with a descent annotated $(j,k)$ for some $k$
(note that $k$ may take $C_j$ distinct possible values by definition).
Therefore, a record $R\in\mathcal{B}_{j}$ ends with a last up-step and
a last descent of length $s_j$. The subpath $R'$ obtained from $R$ by
removing the last up-step and the last descent belongs to
$\mathcal{R}_{s_j-1}$. Hence, the generating function of
$\mathcal{B}_{j}$ is
$B_{j}(y)=R_{s_j-1}(y) \times y C_j = y^{s_j}C_jB(y)^{s_j}$.
Therefore, since a record $R\in \mathcal{B}$ is either empty (i.e. of
size 0) or ends with a descent annotated $(j,k)$, we have:

\begin{equation}\label{eqn:B+B}
  B(y) = 1+ \sum_{j\in\type} B_j(y) = 1+ \sum_{j\in\type} C_j y^{s_j} B(y)^{s_j}
\end{equation}
We are now ready to prove the following theorem.
\begin{theorem}\label{thm:nb_of_rec}
  Algorithm~\gc produces at most $o\paren{\paren{\frac{Q(x)}x}^t}$
  distinct records with $t$ \algo{Color} lines where
  $Q(x) = 1+ \sum_{j\in\type} C_j x^{s_j}$ and for any
  $x\in]0,1]$.

\end{theorem}
In practice, our aim is to minimize the value of
$\frac{Q(x)}x$. Observe that:

\begin{remark}\label{rem:root} In Theorem~\ref{thm:nb_of_rec}, the minimum value of
  $\frac{Q(x)}x$ is as follows:
  \begin{itemize}
  \item If $s_j=1$ for all $j\in\type$, then the minimum is reached for $x=1$
    and $\displaystyle \frac{Q(x)}x =1 + \sum_{j\in\type} C_j$.
  \item Otherwise, the minimum is reached for the unique positive root
    of the polynomial $\displaystyle P(x)= -1 +\sum_{j\in\type}(s_j-1)
    C_j x^{s_j}$.
  \end{itemize}
\end{remark}

\begin{proofof}{Theorem~\ref{thm:nb_of_rec}}
  Let $\displaystyle\lambda=\min_{0< x \le 1}\frac{Q(x)}x$. 
  Let us prove that Algorithm~\gc produces at most $o(\lambda^t)$
  distinct records:
  it suffices to bound $r_t$ (the number of records of size $t$
  of $\mathcal{R}$) by $o(\lambda^t)$. 
  
  If $s_j=1$ for all $j\in\type$, then
  $b_t = \left( \sum_{j\in \type} C_j \right)^t = \left(\lambda -1
  \right)^t$
  by Equation~(\ref{eqn:B+B}).  It follows that
  $r_t = \sum_{0\le \ell \le n} {t \choose \ell} \left( \lambda -1
  \right)^{t-\ell}$
  for sufficiently large $t$ by Equation~(\ref{eqn:R+B}). Finally,  
  $r_t < (n+1)t^{n+1} \left( \lambda -1 \right)^{t}$ and therefore
  $r_t = o(\lambda^t)$.

  From now on, we consider the case where $s_j \ge 2$ for some
  $j\in \type$. As observed by Esperet and
  Parreau~\cite[Lemma~6]{EP12}, there is a constant $C$ (depending
  only on the lengths of the descents) such that $r_t \le b_{t+C}$. It
  suffices hence to show that $b_t = o(\lambda^t)$. For that purpose we
  make use of the \emph{smooth implicit-function schema}\footnote{The
    smooth implicit-function schema is given in \ref{sec:meir_moon}.}
  (SIFS for short) of Meir and Moon~\cite{MM89} (see also Flajolet and
  Sedgewick's book~\cite[Section VII.4.1]{FS08}). 
  Function $B(y)$ does not satisfy the SIFS and we thus
  introduce the function $A(y)$ defined by $A(y) = B(y^\frac1d)-1$
  where $d = \gcd\{s_j\ |\ j\in\type \}$.  We prove in the following
  that $A(y)$ satisfies the SIFS. Note that the size of Dyck
  paths of $\mathcal{B}$ is multiple of $d$. Therefore, we have:
  $$B(y) = \sum_{t\ge 0}b_ty^t = \sum_{t\; \mbox{\scriptsize multiple of} \; d}b_ty^t.$$
  Thus $B(y^\frac1d) = 1+\sum_{t\ge 1}b_{dt}y^t$. Hence
  $A(y)=\sum_{t\ge 0} a_t y^t$ with $a_0=0$ and $a_t=b_{dt}$ for
  $t\ge 1$. Thus $A(y)$ is analytic at 0, $a_0=0$, and $a_t\geq 0$
  for all $t\geq 0$. Furthermore, note that for any sufficiently large
  $t$, the integer $dt$ can be written as a sum which summands belong
  to $\{s_j \ | \ j\in\type\}$. Hence $a_t=b_{dt} >0$ for any
  sufficiently large $t>0$. It follows that $A(y)$ is
  aperiodic\footnote{\emph{Aperiodic} is used in the usual sense of
    Definition~IV.5 of Flajolet-Sedgewick's
    book~\cite{FS08}. Equivalently, there exist three indices
    $i<j<k$ such that $a_ia_ja_k\neq 0$ and $\gcd(j-i,k-i)=1$.}.  By
  Equation~(\ref{eqn:B+B}), we have $A(y) = G(y,A(y))$ for the
  bivariate function $G$ defined by
  $$G(y,z) = \sum_{j\in\type} C_j y^{s_j/d} \left(z+1\right)^{s_j}.$$ 
  Observe that
  $$G(y,z) = \sum_{j\in\type} \sum_{0\le i\le s_j} {s_j \choose i} C_j
  y^{s_j/d} z^{i},$$
  and hence $G(y,z)$ is a bivariate power series satisfying the
  following conditions: 
  \begin{enumerate}
  \item[$(a)$] $G(y,z)$ is analytic in the domain $|y| < +\infty$ and
    $|z| < +\infty$.
  \item[$(b)$] Setting $G(y,z) = \sum_{m,n \ge 0} g_{m,n} y^{m} z^{n}$,
    the coefficients of $G$ satisfy $g_{m,n} \geq 0$, ${g_{0,0}=0}$,
    $g_{0,1}=0$, and $g_{\frac{s_j}{d},s_j}>0$ for the $j\in\type$ such
    that $s_j \geq 2$.
  \item[$(c)$] There exist two positive numbers $r$ and $s$ satisfying the
    system of equations\footnote{ $G_z$ denotes the derivative of $G$ with respect to its second
      variable.}
    $$G(r, s)=s~ \textrm{ and } ~G_z(r, s)=1.$$ 

    Indeed, by setting $X=r^{1/d}(s+1)$, these two equations
    respectively become
    $$\sum_{j\in \type} C_j X^{s_j} =s~ \textrm{ and } ~\sum_{j\in \type} s_j C_j X^{s_j} =s+1.$$ 
    By substracting the first one to the second one, we obtain that $X$ is the
    unique positive root of $P(x)$ (see Remark~\ref{rem:root})  which exists. The first equation hence
    clearly defines $s$. In this first equation adding 1 to both sides,
    and then multiplying them both by $r^{1/d}$, one obtains that $r=
    \left( X / Q(X) \right)^d$.
  \end{enumerate}
  Hence $A(y)=\sum_{t\ge 0} a_t y^t$ satisfies a \emph{smooth
    implicit-function schema} with \emph{characteristic system}
  $(r,s)$, see Definition~\ref{def:SIFS} of \ref{sec:meir_moon}.
  By Theorem~\ref{thm:meir_moon}, we have that $a_t=O\left(
    t^{-\frac32}r^{-t} \right)$. It follows that $a_t=o\left( r^{-t}
  \right)$ and $b_t=o\left( r^{-t/d} \right)=o\left(
    \paren{\frac{Q(X)}{X}}^t \right)$. As $X$ is the unique positive
  root of $P(x)$, this concludes the proof.~\end{proofof}

% Local variables:
% mode: latex
% TeX-master: "0_main.tex"
% ispell-local-dictionary: "english"
% End:

   %% section 4
\section{Some applications of the method to graph coloring problems}
\label{sec:application}

In this section, we apply the framework described in
Section~\ref{sec:general_la_vraie} to different coloring problems. We
improve several known upper bounds by at least an additive constant
and sometimes also by a constant factor. More importantly, this
framework allows simpler proofs with only few
calculations. Indeed, directly using Theorem~\ref{th:main-final}, one avoids the
calculations made in Section~\ref{sec:general}.

\subsection{Non-repetitive coloring}
In a vertex (resp. edge) colored graph, a \emph{$2j$-repetition} is a
path on $2j$ vertices (resp. edges) such that the sequence of colors
of the first half is the same as the sequence of colors of the second
half.  A coloring with no $2j$-repetition, for any $j\ge 1$, is called
\emph{non-repetitive}.  Let $\pi(G)$ be the \emph{non-repetitive
  chromatic number} of $G$, that is the minimum number of colors
needed for any non-repetitive vertex-coloring of $G$. By extension,
let $\pi_l(G)$ be the \emph{non-repetitive choice number} of $G$.
These notions were introduced by Alon {\em et al.} \cite{AGHR02}
inspired by the works on words of Thue \cite{Thu06}.  See~\cite{Gry07}
for a survey on these parameters. Dujmovi\'c {\em et al.}~\cite{DJ+12}
proved that every graph $G$ with maximum degree $\Delta$ satisfies
$\pi_l(G) \le \left\lceil \left( 1 + \frac{1}{\Delta^{\frac13}-1} +
    \frac{1}{\Delta^{\frac13}}\right) \Delta^2 \right \rceil =
\Delta^2 + 2\Delta^{\frac53}+O(\Delta^{\frac43})$
colors. However, their technique could provide tighter bounds from the
second term on~\cite{Jor14}. Here, we provide a simple and
short proof of the following bound.

\begin{theorem}
  Let $G$ be a graph with maximum degree $\Delta\ge 3$. We have:
  $$\pi_l(G) \le
  \ceil{\Delta^2 + \frac{3}{2^{\frac23}}
    \Delta^{\frac53} + \frac{2^{\frac23} \Delta^{\frac53}}{
      \Delta^{\frac13}-2^{\frac13}}} = \Delta^2 +
  \frac{3}{2^{\frac23}} \Delta^{\frac53} + O(\Delta^{\frac43})
  \qquad (\mbox{Note that }\frac3{2^{\frac23}} \approx 1.89)$$
\end{theorem}

\begin{proof}
  To do this, let us use the framework as follows. Let $G$ be any
  graph with maximum degree $\Delta$, and let $n$ denote its number of
  vertices. In this application, the sets $\mathbb{F}(v)$ are closed
  upward.  We directly proceed to the description of the bad events
  $\mathbb{B}(v)$ and the description of the
  required functions. Then, from the set $\mathbb{B}(v)$, we define the
  set $\mathbb{F}(v)$ as its upward closure.  

  \begin{itemize}
  \item Let $\prec$ be any total order on the vertices of
    $G$. $\NUV(\ovphi)$ returns the first uncolored vertex according to
    $\prec$.

  \item Let $\mathbb{B}(v)$ be the set of bad events $\varphi$ anchored
    at $v$ such that vertex $v$ belongs to a repetition in
    $\varphi$. The set $\mathbb{B}(v)$ is partitioned into subsets
    $\mathbb{B}_j(v)$, for $1\le j\le n/2$, in such a way that in every
    $\varphi \in \mathbb{B}_j(v)$ the vertex $v$ belongs to a
    $2j$-repetition.
    Let $\class_j(v)$ be the set of $2j$-vertex paths going through
    $v$. Each set $\mathbb{B}_j(v)$ is partitioned into subsets
    ${\mathbb B}_j^P(v)$ according to the path $P\in \class_j(v)$
    supporting the repetition. If in a bad event $\varphi \in
    \mathbb{B}(v)$ the vertex $v$ belongs to several repetitions, then one of
    the repetitions is chosen arbitrarily to set the value $j$ and the
    path $P$ such that $\varphi\in\mathbb{B}_j^P(v)$. Let
    $C_j=j\Delta^{2j-1}$ as this upper bounds $|\class_j(v)|$. Indeed,
    there are $\Delta^{2j-1}$ possible paths on $2j$ vertices where $v$
    has a given position, and $2j$ possible positions for $v$, but in
    that case every path is counted twice.
  \end{itemize}

  Let us prove that any partial allowed coloring $\varphi$ is a non-repetitive
  coloring. We proceed by induction on the number of colored vertices
  of $\varphi$. If there is no colored vertex, then $\varphi$ is
  clearly non-repetitive. Otherwise, there exists a colored vertex $v$
  such that $\varphi\not\in\mathbb{F}(v)$ and uncoloring $v$ leads to
  a partial allowed coloring $\varphi'$. By induction, $\varphi'$ is
  non-repetitive. Thus, if $\varphi$ contains a repetition, then $v$ is
  necessarily involved. In such a case, we would have
  $\varphi\in\mathbb{F}(v)$, a contradiction.

  \begin{itemize}
  \item The function $\USBE_j(v,\ovphi,P)$ outputs the half of $P$
    containing $v$, and thus $s_j=j$. By Lemma~\ref{lem-USBE}, this
    function fulfills all the requirements.

  \item Given $P$ and the sequence of colors of one half of $P$ (which
    is colored in $\varphi'$), it is easy to recover the sequence of
    colors of the other half of $P$, and so
    $\RBE_j(v,X,P,\varphi')$ is well-defined.
  \end{itemize}
  Consider now
  \begin{eqnarray*}
    Q(x) = 1+\sum_{1\le j \le n/2} C_jx^{s_j}& = & 1+\sum_{1\le j \le n/2} j\Delta^{2j-1} x^j\\
                                             & < & 1 + \frac{\Delta x}{(\Delta^2 x - 1)^2} \mbox{\hspace{1cm} if }\ \ x < \frac1{\Delta^2}
  \end{eqnarray*}
  By setting $X=\frac{1}{\Delta^2}-\left(\frac{2}{\Delta^7}\right)^{\frac13}$ ($X>0$ as $\Delta\ge 3$), one obtains that
  \begin{eqnarray*}
    \frac{Q(X)}{X} & < & \Delta^2 +
                         \frac{3}{2^{\frac23}} \Delta^{\frac53}  + \frac{2^{\frac23}
                         \Delta^{\frac53}}{ \Delta^{\frac13}-2^{\frac13}} 
  \end{eqnarray*}
  By Theorem~\ref{th:main-final}, $G$ admits an allowed
  coloring (hence a non-repetitive coloring) with $\ceil{Q(X)/X}$ colors.
  This concludes the proof of the theorem.
\end{proof}

\bigskip
An edge-coloring is called \emph{non-repetitive} if, for every path
with an even number of edges, the sequence of colors of the first half
differs from the sequence of colors of the second half.  The minimim
number of colors needed to have such a coloring on the edges of $G$ is
called the \emph{Thue index} of $G$, and is denoted by $\pi'(G)$. By
extension, let $\pi'_l(G)$ be the \emph{Thue choice index} of $G$. Alon {\em
  et al.}~\cite{AGHR02} proved that every graph $G$ with maximum
degree $\Delta$ satisfies $\pi'(G) \le c \Delta^2$ with
$c=2e^{16}+1$. We can prove:

\begin{theorem}
  Let $G$ be a graph with maximum degree $\Delta\ge 3$. Then
  $$
  \pi'_l(G)\le \Delta^2 + 2^{\frac43} \Delta^{\frac53} + O(\Delta^{\frac43}).
  $$
\end{theorem}
The only difference with the vertex case is that
$C_j=2j\Delta^{2j-1}$.

\subsection{Facial Thue vertex-coloring}

We consider in this subsection a slight variation of non-repetitive
coloring which applies to plane graphs (i.e. embedded planar
graphs). Here the restriction on repetitions only applies on facial
paths. More formally, consider a plane graph $G$. A {\em facial path}
of $G$ is a path on consecutive vertices on the boundary walk of some
face of $G$. A vertex-coloring of $G$ is said to be {\em facially
  non-repetitive} if none of the facial paths is a repetition. The
notion can be extended to list coloring. Let $\pi_{f}(G)$
(resp. $\pi_{fl}(G)$) denote the {\em facial Thue chromatic number}
(resp. {\em facial Thue choice number}) that is the minimum integer
$k$ such that $G$ is facially non-repetitively $k$-colorable
(resp. $k$-choosable).  Bar\'at and Czap~\cite{BC13} proved that for
any plane graph $G$, $\pi_f(G)\le 24$.  Whether the facial Thue choice
number of plane graphs could be bounded from above by a constant is
still an open question. Recently Przyby\l o {\em et al.}  \cite{PSS13}
proved that, if $G$ is a plane graph of maximum degree $\Delta$, then
$\pi_{fl}(G) \le 5\Delta$, and asymptotically, $\pi_{fl}(G)\le
(2+o(1)) \Delta$. We improve these upper bounds as follows:

\begin{theorem}
  Let $G$ be a plane graph with maximum degree $\Delta\ge 2$. Then,
  $$
  \pi_{fl}(G) \le \ceil{ \Delta + 4\sqrt{\Delta} +3 }
  $$
\end{theorem}
\begin{proof} Let $G$ be a plane graph with maximum degree $\Delta$.
  In this application, the sets $\mathbb{F}(v)$ are closed upward.  We
  directly proceed to the description of the bad events $\mathbb{B}(v)$
  and the description of the required functions. Then, from the set
  $\mathbb{B}(v)$, we define the set $\mathbb{F}(v)$ as its upward
  closure.
  \begin{itemize}
  \item As previously, let $\prec$ be any total order on the vertices of
    $G$. $\NUV(\ovphi)$ returns the first uncolored vertex according to
    $\prec$.
  \item For $1\le j \le \floor{n/2}=p$, let $\mathbb{B}_j(v)$ be the set
    of bad events $\varphi$ such that vertex $v$ belongs to a repetition
    on a facial $2j$-vertex path $P$. Let $\class_j(v)$ be the set of
    facial $2j$-vertex paths going through $v$. Each set
    $\mathbb{B}_j(v)$ is partitioned into sets $\mathbb{B}^P_j(v)$, for
    every $P\in \class_j(v)$, according to the path $P$ supporting the
    repetition. The number of obtained classes is such
    that we set $C_1= \Delta$ and $C_j= 2j\Delta$ for $j\ge 2$.
    Indeed, there are at most $\Delta$ possible faces for containing $P$, and
    $2j$ positions for $v$ in $P$.
  \end{itemize}
  Let us prove that any partial allowed coloring $\varphi$ is a facial
  non-repetitive coloring.  Proceed by induction on the number of
  colored vertices of $\varphi$. Either $\varphi$ has no colored
  vertex and it is facially non-repetitive, or there exists a colored
  vertex $v$ such that $\varphi\not\in\mathbb{F}(v)$ and uncoloring
  $v$ leads to a partial allowed coloring $\varphi'$, that is hence
  facial non-repetitive. Thus, if $\varphi$ contains a facial
  repetition, then $v$ is necessarily involved. In such a case, we would
  have $\varphi\in\mathbb{F}(v)$, a contradiction.

  \begin{itemize}
  \item The function $\USBE_j(v,\ovphi,P)$ outputs the half of the path
    $P$ containing $v$, and thus $s_j=j$. By Lemma~\ref{lem-USBE}, this
    function fulfills all the requirements.
  \item Given $P$ and the sequence of colors of the colored half of $P$,
    it is easy to recover the sequence of colors of the uncolored half
    of $P$, and so $\RBE_j(v,X,P,\varphi')$ is well-defined.
  \end{itemize}
  Consider now
  \begin{eqnarray*}
    Q(x) = 1+\sum_{1\le j \le n/2} C_jx^{s_j}& = & 1+\Delta x+ \sum_{2\le j \le p} 2j\Delta x^j\\
                                             & < & 1 + \Delta x + 2 \Delta x^2 \frac{2-x}{(x-1)^2} \quad \mbox{\hspace{1cm} for } x<1\\
  \end{eqnarray*}
  By setting $X=\frac{1}{2\sqrt{\Delta}}$, and as $\Delta \ge 2$ one obtains that
  \begin{eqnarray*}
    \frac{Q(X)}{X} & < & \Delta + 4\sqrt{\Delta} +3\\
  \end{eqnarray*}
  By Theorem~\ref{th:main-final}, $G$ admits an allowed
  coloring (hence a facial non-repetitive coloring) with $\ceil{Q(X)/X}$ colors.
  This concludes the proof of the theorem.
\end{proof}

Piotr Micek recently announced that this theorem can be improved
asymptotically as for any plane graph $G$,
$\pi_{fl}(G) \le O(\log \Delta)$~\cite{Jor14}.

\subsection{Facial Thue edge-coloring}\label{subsec:thue-edge}

Consider the {\em facial Thue choice index} $\pi'_{fl}(G)$ of a plane
graph $G$, that is the minimum integer $k$ such that $G$ is facially
non-repetitively edge $k$-choosable. Schreyer and
\v{S}krabul'\'akov\'a~\cite{SS12} proved that plane graphs have bounded facial
Thue choice index, more precisely $\pi'_{fl}(G) \le 291$. Recently
Przyby\l o~\cite{P13} improved that bound to 12. To obtain that upper
bound with our framework, it is sufficient to consider as bad events
the partial colorings having a facial $2j$-repetition (for any $j\ge
1$) with costs $C_j=4j$ since an edge belongs to at most $4j$ facial
$2j$-edge paths.

Let us explain a way to improve that upper bound. The idea is that at
each step the algorithm chooses the edge $e$ to be colored in such a way
that $e$ is facially adjacent to an uncolored edge $e'$. Therefore, if
at some step the algorithm colors such an edge $e$, then this edge belongs
to at most $1+2j$ facial $2j$-edge paths going through colored edges
(one path in the face incident to $e$ and $e'$ and $2j$ paths on the
other face incident to $e$).  However, such an edge $e$ does not
always exist. For example if the algorithm has colored all the graph
$G$ but one edge, then this edge may belong to $4j$ colored facial
$2j$-edge paths. We manage to use this trick to obtain the improved
bound of 10.

We will need the following definition. Given a plane graph $G$, its
{\em medial graph} $M(G)$ is defined as follows:
\begin{itemize}
\item its vertex set is the set of edges of $G$;
\item there is an edge $uv$ between the vertices $u$ and $v$ of $M(G)$
  if and only if the corresponding edges in $G$ are facially adjacent
  (i.e. adjacent and both incident to the same face).
\end{itemize}

\begin{theorem}\label{thm:thue-edge-sauf-une}
  For any plane graph $G$, any edge $e^*$ of $G$, and any assignment of
  lists of size $9$, there exists a partial facial Thue edge-coloring of
  $G$ where all the edges except $e^*$ are colored. 
\end{theorem}

\begin{proof}
  Let $G$ be a plane graph with maximum degree $\Delta$, and let $e^*$
  be any edge of $G$. In this application, we want to ensure that at
  each iteration of the main loop the current edge to color is
  facially adjacent to (at least) one uncolored edge. This leads us to
  sets $\mathbb{F}(e)$ that are not closed upward. Hence they need to
  be described with care. For a given edge $e$, the set
  $\mathbb{F}(e)$ contains the partial colorings with a facial
  repetition involving $e$, and the partial colorings where the set of
  uncolored edges (i.e. vertices of $M(G)$), including $e^*$, induces
  a disconnected graph in $M(G)$. Hence the set of allowed colorings
  is the set of partial colorings with no facial repetition, and where
  uncolored edges, including $e^*$, induce a connected graph in
  $M(G)$.

  We conveniently define $\NUV$ in order to avoid bad events dealing
  with the case where uncolored edges induce a disconnected graph in
  $M(G)$. 
  \begin{itemize}
  \item For any set $X\subseteq E(G)$ such that $e^*\in X$, and such
    that $M(G)[X]$ is connected, the edge \linebreak $e=\NUV(E(G)\setminus
    X)$ must be such that $M(G)[X-e]$ is connected. Hence, $e$ may
    be chosen among leaves of a spanning tree of $M(G)[X]$
    rooted at $e^*$.
  \end{itemize}
  Hence with that definition of $\NUV$ we have that for a given edge
  $e$, the set of bad events $\mathbb{B}(e)$ contains the partial
  colorings with a facial repetition involving $e$, where $e$ is
  facially adjacent to an uncolored edge $e'$ (its parent in the
  spanning tree described above, which might be $e^*$), and where the
  set of uncolored edges induces a connected graph in $M(G)$. Let us
  introduce the bad event types and classes:
  \begin{itemize}
  \item For $1\le j \le p=\floor{n/2}$, let $\mathbb{B}_j(e)$ be the set
    of bad events anchored at $e$ such that $e$ has an uncolored
    facially adjacent edge $e'$, and $e$ belongs to a repetition on a
    (colored) facial $2j$-edge path $P$. 
    
    The partition into classes is not obvious. Let $e_1, e_2, e_3$ and
    $e_4$ be the (at most four) edges of $G$ facially adjacent to $e$,
    and let $e'\in\{e_1,e_2,e_3,e_4\}$ be the uncolored one with
    smallest index. Let us now partition $\mathbb{B}_j(e)$ into sets
    $\mathbb{B}_j^{e',P}(e)$ according to the uncolored edge $e'$ and
    the path $P$ supporting the repetition. We have seen earlier that
    given an edge $e'$ there are at most $1+2j$ possible paths $P$. As
    there are up to four possibilities for $e'$ this partition has
    $4+8j$ parts, but the cases where $e'$ has distinct values are 
    independent. Let us hence merge these parts as follow. Let
    $\mathbb{B}_j^k(e)$, for $1\le k\le 1+2j$, be the union of
    $\mathbb{B}_j^{e_1,P_1}(e)$, $\mathbb{B}_j^{e_2,P_2}(e)$,
    $\mathbb{B}_j^{e_3,P_3}(e)$ and $\mathbb{B}_j^{e_4,P_4}(e)$, for
    some choice of paths $P_1$, $P_2$, $P_3$ and $P_4$. The obtained
    partition has $C_j=1+2j$ classes.

  \item Given the set of colored edges $\ovphi$ of some bad event
    $\varphi \in \mathbb{B}_j(e)$, one can determine the facially
    adjacent uncolored edge $e'$. Hence given (also) the class $k$ such
    that $\varphi \in \mathbb{B}_j^k(e)$, one can determine the path $P$
    supporting the repetition.  The function $\USBE_j(e,\ovphi,k)$
    outputs the half of the path $P$ containing $e$, and thus
    $s_j=j$. Note that as the edges of $P$ are incident to the same
    face, and as $e$ and $e'$ are facially adjacent, uncoloring this set
    of edges leads to a partial coloring that has no repetition and such
    that the uncolored edges induce a connected graph in $M(G)$, hence an
    allowed coloring (as required).
  \item Using again the fact that $P$ can be retrived from $\ovphi$
    ($=X$ here) and $k$, one can easily design a function
    $\RBE_j(v,X,k,\varphi_{i+1})$.
  \end{itemize}
  Consider now
  \begin{eqnarray*}
    Q(x) = 1+\sum_{1\le j \le n/2} C_jx^{s_j}& = & 1+\sum_{1\le j \le n/2} (1+2j) x^j\\
                                             & < & \frac{1}{1-x}  + \frac{2x}{(1-x)^2} \mbox{\hspace{1cm} if }\ \ x < 1\\
  \end{eqnarray*}
  By setting $X=\frac{\sqrt{17}-3}{4}$, one obtains that $Q(X)/X < 9$.
  Hence by Theorem~\ref{th:main-final}, $G$ admits a partial allowed
  9-coloring (hence a partial facial Thue edge-coloring) where $e^*$
  is the onlyuncolored edge.  This concludes the proof of the theorem.
\end{proof}

Given Theorem~\ref{thm:thue-edge-sauf-une}, it seems likely that
$\pi'_{fl}(G) \le 9$ for any plane graph $G$. Actually one can show
that it is the case if $G$ has an edge $e^*$ incident to two faces of
small sizes.  Unfortunately we do not achieve this bound here, but we
prove:

\begin{corollary}\label{cor:thue-edge}
  For any plane graph $G$, $\pi'_{fl}(G) \le 10$.
\end{corollary}
\begin{proof}
  For a given $G$, pick an arbitrary edge $e^*\in E(G)$ and an arbitrary
  color $c\in L(e^*)$.  For all the other edges of $G$, remove color $c$
  from their list. Now all these lists have size at least 9.  By
  Theorem~\ref{thm:thue-edge-sauf-une}, it is possible to color all the
  edge of $G$ except $e^*$, avoiding facial repetitions.  Then
  coloring $e^*$ with $c$ cannot create any repetition, as $c$ does
  not appear elsewhere in $G$.
\end{proof}

\begin{remark}
  Note that in the proof of Theorem~\ref{thm:thue-edge-sauf-une} we
  only use the fact that edges are adjacent to at most two faces, and
  thus it extends to any graph embedded on any surface. Hence,
  Theorem~\ref{thm:thue-edge-sauf-une} and
  Corollary~\ref{cor:thue-edge}, both extend to arbitrary surface.
\end{remark}

\subsection{Generalised acyclic coloring}

Let $r\ge 3$ be an integer. An $r$-acyclic vertex-coloring is a proper
vertex-coloring such that every cycle $C$ uses at least $\min(|C|,r)$
colors. This generalisation of the notion of acyclic coloring (the
$r=3$ case) was introduced by Gerke {\em et al.} in the context of
edge-coloring \cite{GGW06} and then by Greenhill and Pikhurko in the
context of vertex-coloring \cite{GP05}. Let $A_r(G)$ be the minimum
number of colors in any $r$-acyclic vertex-coloring of $G$. By
extension, let $A_r^l(G)$ be the $r$-acyclic choice number of
$G$. Greenhill and Pikhurko \cite{GP05} proved in particular that, for
$r\ge 4$ and $\Delta\ge 3$, every graph $G$ with maximum degree
$\Delta$ satisfies $A_r(G)\le c \Delta^{\lfloor r/2 \rfloor}$ where
$c=2^{(r+2)/3}r(r+2)$.  We reduce this constant factor as follows.

\begin{theorem}
  Let $G$ be a graph with maximum degree $\Delta\ge 3$. For any $r\ge 4$,
  we have that $A_r^l(G)\le \Delta^{\floor{r/2}} + O\left(
    \Delta^{(r+1)/3} \right)$.
\end{theorem}

In the following, all the defined events are strongly inspired by
those defined by Greenhill and Pikhurko \cite{GP05}. Let $G$ be any
graph with maximum degree $\Delta$, and let $n$ denote its number of
vertices. Let $\prec$ be any total order on the vertices of
$G$. $\NUV(\ovphi)$ returns the first uncolored vertex according to
$\prec$. In this application, the sets $\mathbb{F}(v)$ are closed
upward. We hence use Lemma~\ref{lem-USBE}, to ensure that each function
$\USBE$ fulfills all the requirements. We proceed now to the
description of the bad events (the sets $\mathbb{F}(v)$ being deduced
from $\mathbb{B}(v)$), and the description of the required
functions. We distinguish two cases according to $r$'s parity.

\subsubsection{Case $r$ even}

Set $r=2\ell$ with $\ell\ge 2$.  We consider the following sets of bad
events anchored at vertex $v$:

\begin{itemize}
\item Let $\mathbb{B}_1(v)$ be the set of bad events $\varphi$ where
  ``there exists a vertex $u$ at distance at most $\ell$ (from $v$)
  having the same color as $v$''. Let $\class_1(v)$ be the set of
  vertices $u$ at distance at most $\ell$ from $v$. As $|\class_1(v)|
  \le \sum_{i=1}^{\ell}\Delta(\Delta-1)^{i-1} =
  \frac{\Delta((\Delta-1)^\ell - 1)}{\Delta-2} \le \Delta^{\ell}$ we
  set $C_1=\Delta^{\ell}$. Each set $\mathbb{B}_1(v)$ is partitioned
  into classes $\mathbb{B}_1^u(v)$, for every vertex $u \in
  \class_1(v)$, according to the vertex $u$ that is colored like $v$.
  $\USBE_1(v,\ovphi,u)$ outputs the vertex $v$, and thus $s_1=1$. In
  addition, $\RBE_1(v,X,u,\varphi')$ outputs the partial coloring
  $\varphi$ obtained from $\varphi'$ by coloring $v$ with color
  $\varphi'(u)$.
\end{itemize}
Here it is clear that an allowed coloring is a distance $\ell$ proper
coloring. Furthermore, as $r=2\ell$, cycles $C$ of length at most
$r+1$ will receive $|C|$ distinct colors.

\begin{itemize}
\item Let $\mathbb{B}_2(v)$ be the set of bad events $\varphi$ where
  ``$v$ belongs to a path $P$ on $r+2$ vertices such that $v$ and two
  other colored vertices, say $a$, $b$, have colors that already
  appear on $P \setminus \{v,a,b\}$''. Let us define a partition of
  $\mathbb{B}_2(v)$. Consider the set $\class_2(v)$ formed by all
  tuples $(P,a,b,v',a',b')$ such that $P$ is a path on $r+2$ vertices
  containing vertices $v,a,b,v',a',b'$ where $|\{v,a,b\}|=3$, $1\le
  |\{v',a',b'\}|\le 3$ and $\{v,a,b\}\cap \{v',a',b'\}=\emptyset$. Let
  $\mathbb{B}_2^{(P,a,b,v',a',b')}(v) \subset \mathbb{B}_2(v)$ be the
  class of bad events $\varphi$ where ``both $v$ and $v'$ have the
  same color, both $a$ and $a'$ have the same color, and both $b$ and
  $b'$ have the same color''. Let us count the number of such
  classes. First observe that $v$ belongs to at most
  $\frac{r+2}{2}\Delta(\Delta-1)^{r}$ paths on $r+2$ vertices. Now
  observe that there are at most $r+2$ possible choices for each
  vertex $a,b,v',a',b'$. Hence let us set
  $C_2=\frac{1}{2}(r+2)^6\Delta^{r+1}$. $\USBE_2(v,\ovphi,(P,a,b,v',a',b'))$
  outputs the set $\{v,a,b\}$, and thus $s_2=3$. In addition,
  $\RBE_2(v,X,(P,a,b,v',a',b'),\varphi')$ outputs the partial coloring
  $\varphi$ obtained from $\varphi'$ by coloring
  vertices $v$, $a$ and $b$ respectively with colors
  $\varphi'(v')$, $\varphi'(a')$ and $\varphi'(b')$.
\end{itemize}
These bad events imply that in an allowed coloring, cycles of length
at least $r+2$ contain at least $r$ colors. Hence an allowed coloring
is also a generalised $r$-acyclic coloring. Consider now
\begin{eqnarray*}
  Q(x) = 1+\sum_{1\le j \le n/2} C_jx^{s_j}& = & 1+C_1 x+ C_2 x^3\\
\end{eqnarray*}
By setting $X=\left(\frac{1}{2 C_2}\right)^{\frac13}$ one obtains that
\begin{eqnarray*}
  \frac{Q(X)}{X} & = & C_1 + \frac{3}{2^{\frac23}}C_2^{\frac13}\\
                 & = & \Delta^{\ell} + \frac32 (r+2)^2\Delta^{(r + 1)/3} 
\end{eqnarray*}
By Theorem~\ref{th:main-final}, $G$ admits an allowed
coloring (hence  a generalised $r$-acyclic coloring) with $\ceil{Q(X)/X}$ colors.
This concludes the proof of the theorem for $r$ even.

\subsubsection{Case $r$ odd}

The odd case is similar to the even case. Let $r=2\ell+1$ with $\ell
\ge 2$.  Let us use again the two types of bad events defined
above. Now, type 1 bad events are sufficient to deal with cycles of
length at most $r$. Type $2$ bad events are still sufficient to deal
with cycles of length at least $r+2$. It remains to deal with cycles
of length $r+1$.  Type 1 bad events forbid some kinds of length $r+1$
cycles. As $r+1=2\ell+2$, the cycles of length $r+1$ that are not
forbidden by type 1 bad events are those where each color appears only
once, or where colors appearing several times, do it on antipodal
vertices.  We thus add two other bad event types to deal with this
kind of cycles of length $r+1$.

A pair of vertices $\{u,u'\}$ is said to be \emph{special} if $u$ and
$u'$ are at distance exactly $\ell +1$ and if there exist at least
$\Delta^{\frac{\ell+1}{3}}$ paths of length $\ell+1$ linking $u$ and
$u'$. Consider the two following new sets of bad events:
\begin{itemize}
\item Let $\mathbb{B}_3(v)$ be the set of bad events $\varphi$ where
  ``there exists a special pair $\{v,u\}$ such that $v$ and $u$ have
  the same color''. Let $\class_3(v)$ be the set of vertices $u$ such
  that $\{v,u\}$ is a special pair. Each set $\mathbb{B}_3(v)$ is partitioned
  into classes $\mathbb{B}_3^{u}(v)$ according to the vertex $u$ colored like $v$. As there
  are at most $\Delta^{\ell+1}$ paths of length $\ell+1$ starting from
  $v$, there exist at most $\Delta^{\frac23 (\ell+1)}
  =\Delta^{(r+1)/3} = C_3$ such classes. Functions $\USBE_3$ and
  $\RBE_3$ are defined similarly to the first type of bad events, with
  $s_3=1$.
\item Let $\mathbb{B}_4(v)$ be the set of bad events $\varphi$ where
  ``$v$ belongs to a cycle $C$ of length $r+1 = 2\ell +2$ such that
  $v$ and its antipodal vertex $v'$ (on $C$) have the same color, are
  at distance $\ell+1$ from each other but do not form a special pair,
  and such that $C$ contains another pair of antipodal vertices
  $\{u,u'\}$ having the same color''.  Let $\class_4(v)$ be the set of
  couples $(C,u)$ such that $C$ is a $(r+1)$-cycle containing $v$ and
  $u$ as non-antipodal vertices.  Each set $\mathbb{B}_4(v)$ is
  partitioned into classes ${\mathbb B}_4^{(C,u)}(v)$, for every
  $(C,u)\in\class_4(v)$. There exist at most $\ell \Delta^{\frac43
    (\ell +1)} = \ell \Delta^{\frac23 (r+1)} = C_4$ such classes.
  Indeed, there are $\Delta^{\ell +1}$ choices for vertex $v'$ and the
  path from $v$ to $v'$; as $v$ and $v'$ do not form a special pair,
  there are $\Delta^{\frac13 (\ell +1)}$ choices for the path from
  $v'$ back to $v$; and finally there are $\ell$ possibilities for the
  pair $\{u,u'\}$ of antipodal vertices. The function
  $\USBE_4(v,\ovphi,(C,u))$ outputs $\{v,u\}$, so $s_4=2$, and
  $\RBE_4$ clearly exists.
\end{itemize}
One can check that these two new types of bad events handle the
remaining cycles of length $r+1$ colored with less than $r$ colors.
This ensures us that allowed colorings are generalised $r$-acyclic
colorings. Consider now
\begin{eqnarray*}
  Q(x) = 1+\sum_{1\le j \le n/2} C_jx^{s_j}& = & 1 + C_1 x + C_2 x^3 + C_3 x + C_4 x^2\\
\end{eqnarray*}
By setting $X=\frac{1}{\Delta^{(r+1)/3}}$ one obtains that
\begin{eqnarray*}
  \frac{Q(X)}{X} & = & \Delta^\ell + \Delta^{(r+1)/3} + \Delta^{(r+1)/3} \left(1+ \ell + \frac12 (r+2)^6 \right)\\
                 & = & \Delta^\ell +  \Delta^{(r+1)/3} \left( 2+\ell + \frac12 (r+2)^6 \right)\\
\end{eqnarray*}
By Theorem~\ref{th:main-final}, $G$ admits an allowed coloring (hence
a generalised $r$-acyclic coloring) with $\ceil{Q(X)/X}$ colors.  This
concludes the proof of the theorem for $r$ odd.

\subsection{Colorings with restrictions on pairs of color classes}
\label{ssec:chi2F}

For many graph colorings, the color classes are asked to induce
independent sets while another property is asked to each pair of color
classes. Aravind and Subramanian~\cite{AS11} introduced a general
definition that captures many known colorings. In their definition,
restrictions may apply to any $\ell$ color classes, for any $\ell\ge
2$. Let us restrict ourselves to the case $\ell=2$.

\medskip

Given a family $\mathcal{F}$ of connected bipartite graphs, a {\em
  $(2,\mathcal{F})$-subgraph coloring} of $G$ is a proper coloring of
$V(G)$ such that the subgraph of $G$ induced by any two color classes
does not contain any isomorphic copy of $H$ as a subgraph, for each
$H\in \mathcal{F}$. Denote by $\chi_{2,\mathcal{F}}(G)$ the minimum
number of colors used by any $(2,\mathcal{F})$-subgraph coloring of
$G$. Denote by $\chi_{2,\mathcal{F}}(\Delta)$ the maximum value of
$\chi_{2,\mathcal{F}}(G)$ for any graph $G$ having maximum degree at
most $\Delta$. For example, when $\mathcal{F}$ is the family of even
cycles, $(2,\mathcal{F})$-subgraph coloring is the usual acyclic
vertex-coloring.

Using random graphs, Aravind and
Subramanian~\cite{AS11} showed the following lower bound on
$\chi_{2,\mathcal{F}}(\Delta)$.
\begin{theorem}[Aravind and Subramanian~\cite{AS11}]\label{thm:AS13}
  Given a connected bipartite graph $H$ with $m$ edges ($m\ge 2$), we have
  $$ \chi_{2,\{H\}}(\Delta) = \Omega\left(
    \frac{\Delta^\frac{m}{m-1}}{(\log \Delta)^{1/(m-1)}}\right) $$
  Hence,
  the same bound applies to $\chi_{2,\mathcal{F}}(\Delta)$ for any
  family $\mathcal{F}$ containing a graph $H$ with $m$ edges.
\end{theorem}
The same authors later showed that this lower bound is almost tight.
Let $m\ge 2$ be an integer and let $\mathcal{F}$ be a family of
connected bipartite graphs such that all the graphs have at least $m$
edges.

\begin{theorem}[Aravind and Subramanian~\cite{AS13}]
  For some constant $C$ depending only on $\mathcal{F}$, we have
  $$\chi_{2,\mathcal{F}}(\Delta) \le C \Delta^\frac{m}{m-1}$$
\end{theorem}

Partition the graphs in $\mathcal{F}$ according to their number of
vertices.  Let $\mathcal{F}_v^{\le m}$ (resp. $\mathcal{F}_v^{>m}$)
denote the subset of $\mathcal{F}$ with graphs on at most $m$ vertices
(resp. more that $m$ vertices).  Let also
$k_v^{\le m}=|\mathcal{F}_v^{\le m}|$. We consider another parition of
$\mathcal{F}$ according to the number of edges in each graph. Let
$\mathcal{F}_e^m$ (resp. $\mathcal{F}_e^{>m}$) denote the subset of
$\mathcal{F}$ with graphs on exactly $m$ edges (resp. more that $m$
edges); and let $k_e^{m}=|\mathcal{F}_e^m|$.

The constant $C$ mentionned in Theorem~\ref{thm:AS13} is either
$64(m+1)^3k_v^{\le m}$ or $128(m+1)^3$ according to whether $k_v^{\le
  m}>0$ or not. Following the approach of Aravind and Subramanian, we
improve $C$ as follows.

\begin{theorem}\label{thm:chi2F}
  We have
  \begin{eqnarray}
    \chi_{2,\mathcal{F}}(\Delta) &<&
                                     (k_v^{\le m}+71)(m+1)\Delta^{\frac{m}{m-1}} \\
    \chi_{2,\mathcal{F}}(\Delta) &<& \left( k_e^{m} + 1 + o(1) \right)(m+1)\Delta^{\frac{m}{m-1}}
  \end{eqnarray}
\end{theorem}

\begin{proof}
  Let us use the framework described in
  Section~\ref{sec:general_la_vraie} as follows. Let
  $\mathcal{F}=\set{H_1,H_2,\ldots}$. Let us also denote by $n_i$
  and $m_i$ the number of vertices and edges in the forbidden graph
  $H_i$ for each $i$ (recall $m_i \ge m$).  For convenience, we
  introduce the value $\gamma = \frac{m}{m-1}$. Let $G$ be any graph
  with maximum degree $\Delta$, and let $n$ denote its number of
  vertices. As in this application, the sets $\mathbb{F}(v)$ are
  closed upward we directly proceed to the description of the bad
  events (as $\mathbb{F}(v)$ is deduced from $\mathbb{B}(v)$), and the
  description of the required functions.

  \begin{itemize}
  \item Let $\prec$ be any total order on the vertices of
    $G$. $\NUV(\ovphi)$ returns the first uncolored vertex according to
    $\prec$.
  \end{itemize}

  \begin{itemize}
  \item Let $\mathbb{B}_{E}(v)$ be the set of bad events $\varphi$
    anchored at $v$ such that vertex $v$ belongs to a monochromatic edge
    $uv$ (in $\varphi$). Let $\class_E(v)=N(v)$. Let us partition
    $\mathbb{B}_{E}(v)$ into classes $\mathbb{B}_{E}^{u}(v)$ according
    to which edge $uv$ is monochromatic in $\varphi$, for
    $u\in\class_E(v)$. Clearly $|\class_E(v)| \le \Delta$, thus let $C_E=\Delta$.
  \end{itemize}
  From here it is clear that an allowed coloring is proper.

  \begin{itemize}
  \item The function $\USBE_E(v,\ovphi,u)$ outputs the singleton $\{v\}$
    and thus $s_E=1$. By Lemma~\ref{lem-USBE}, this
    function fulfills all the requirements.

  \item $\RBE_E(v,X,u,\varphi')$ outputs the partial coloring $\varphi \in
    \mathbb{B}_{E}^{u}(v)$ obtained from $\varphi'$ by coloring $v$ with
    color $\varphi'(u)$.
  \end{itemize}
  Following the approach of Aravind and Subramanian~\cite{AS13}, we
  extend the notion of special pairs introduced by Alon et
  al.~\cite{AMR91} to bigger sets. For any $j\ge 2$, a $j$-set $S$ of
  $G$ (i.e. a set of size $j$) is \emph{special} if the set $X =
  \bigcap_{v\in S} N(v)$ has size greater than $\Delta^{j-\gamma(j-1)}$.
  Let us define the corresponding bad events.
  \begin{itemize}
  \item For $2\le j< n$, let $\mathbb{B}_{\jset}(v)$ be the set of bad
    events $\varphi$ anchored at $v$ such that vertex $v$ belongs to a
    monochromatic special $j$-set $S$. Let $\class_\jset(v)$ be the
    set of special $j$-sets containing $v$.  Let us partition
    $\mathbb{B}_\jset(v)$ into classes $\mathbb{B}_\jset^S(v)$
    according to which special $j$-set $S\in \class_\jset(v)$ is
    monochromatic. By Claim~\ref{cl:51-spec}, the number of classes is
    at most $\frac{1}{(j-1)!}\Delta^{\gamma(j-1)} = C_\jset$.

    \begin{claim}\label{cl:51-spec} 
      Any vertex $v$ of $G$ belongs to less
      than $\frac{1}{(j-1)!}\Delta^{\gamma(j-1)}$ special $j$-sets, for any $j\ge 2$.
    \end{claim}
    \begin{proof} 
      Observe that $v$ belongs to $\Delta {\Delta-1 \choose j-1}$ stars (on
      $j+1$ vertices) centered in $N(v)$ having $j-1$ leaves in $N^2(v)$
      (first choose a center and then $j-1$ of its neighbors). Now the
      $j$ leaves of such a star are contained in at most one special
      $j$-set of $v$. On the other hand, a special $j$-set containing $v$
      covers more than $\Delta^{j-\gamma(j-1)}$ of these stars. Hence $v$
      belongs to less than $\Delta {\Delta-1 \choose j-1} \times
      \Delta^{\gamma(j-1)-j} < \frac{1}{(j-1)!}\Delta^{\gamma(j-1)}$ special
      $j$-sets.
    \end{proof}

  \end{itemize}

  From here it is clear that in an allowed coloring there will be no
  monochromatic special $j$-set.

  \begin{itemize}
  \item For $2\le j< n$, let the function $\USBE_\jset(v,\ovphi,S)$
    outputs a $(j-1)$-subset of $S$ containing $v$ ; 
    thus $s_\jset=j-1$. Again by Lemma~\ref{lem-USBE}, this
    function fulfills all the requirements.
    
  \item If $\RBE_j(v,X,S,\varphi')$ is called, then there is only one
    vertex of $S$ colored in $\varphi'$, say $w$.  Hence
    $\RBE_j(v,X,S,\varphi')$ outputs the partial coloring obtained from
    $\varphi'$ by coloring all the vertices of $S$ with $\varphi'(w)$.
  \end{itemize}
  As proposed in~\cite{AS13}, one bad event type can deal with all the graphs in $\mathcal{F}_v^{>m} \subseteq
  \mathcal{F}$ the set of forbidden graphs having more than $m$ vertices. 
  \begin{itemize}
  \item Let $\mathbb{B}_{\mathcal{F}_v^{>m}}(v)$ be the set of bad
    events $\varphi$ anchored at $v$ such that vertex $v$ belongs to a
    connected properly bicolored subgraph $I$ on $m+1$ vertices. Note
    that such subgraph $I$ of $G$ is not necessarily isomorphic to a
    graph of $\mathcal{F}_v^{>m}$. However this type of bad events
    deal with all the graphs of $\mathcal{F}$ with at least $m+1$
    vertices. Let $\class_{\mathcal{F}_v^{>m}}(v)$ be the set of all
    connected bipartite graphs $I$ on $m+1$ vertices that contain
    vertex $v$. We partition $\mathbb{B}_{\mathcal{F}_v^{>m}}(v)$ into
    classes $\mathbb{B}_{\mathcal{F}_v^{>m}}^I(v)$ according to the
    bicolored subgraph $I$. By the proof of Lemma~2.4 in~\cite{AS11}
    we have that the number of classes,
    $|\class_{\mathcal{F}_v^{>m}}(v)| \le
    (m+1)4^{m+1}\Delta^{m}=C_{\mathcal{F}_v^{>m}}$.
  \end{itemize}
  From here it is clear that in an allowed coloring there will be no properly bicolored
  copy of any $H_i \in \mathcal{F}$ with more than $m$ vertices.

  \begin{itemize}
  \item The function $\USBE_{\mathcal{F}_v^{>m}}(v,\ovphi,I)$ outputs a
    $(m-1)$-subset of $V(I)$ containing $v$ (recall $I$ is a properly
    bicolored subgraph on $m+1$ vertices), such that the two remaining
    vertices $v_1$ and $v_2$ are adjacent (and thus have distinct
    colors). Note that $s_{\mathcal{F}_v^{>m}}=m-1$. Again by Lemma~\ref{lem-USBE},
    this function fulfills all the requirements.

  \item If $\RBE_{\mathcal{F}_v^{>m}}(v,X,I,\varphi')$ is called, then there
    are only two adjacent vertices of $I$, $v_1$ and $v_2$, colored in
    $\varphi'$.  Hence $\RBE_{\mathcal{F}_v^{>m}}(v,X,I,\varphi')$ outputs
    the partial coloring obtained from $\varphi'$ by properly extending
    the 2-coloring of $v_1$ and $v_2$ to the whole $I$.
  \end{itemize}

  We define a new bad event type for each graph $H_i\in
  \mathcal{F}_v^{\le m}$, that is each graph of $\mathcal{F}$ with at
  most $m$ vertices.  Let $V_1$ and $V_2$ be the two independent sets
  partitioning $V(H_i)$.
  \begin{itemize}
  \item Let $\mathbb{B}_{Hi}(v)$ be the set of bad events $\varphi$
    anchored at $v$ such that vertex $v$ belongs to a properly
    2-colored subgraph $S$ isomorphic to $H_i\in \mathcal{F}_v^{\le
      m}$, and such that $S$ does not contain a monochromatic special
    $j$-set. Let $\class_{Hi}(v)$ be the set of all subgraphs $S$
    isomorphic to $H_i$, containing $v$, and without special $j$-set
    contained in one of the two parts of $S$. The set
    $\mathbb{B}_{Hi}(v)$ is partitioned into classes
    $\mathbb{B}_{Hi}^S(v)$ according to the bicolored copy, $S$. By
    Claim~\ref{cl:51-H} (see below), the number of classes is at most
    $ n_i \Delta^{\gamma(n_i-2) -\frac{m_i-m}{m-1}} = C_{Hi}$.
  \end{itemize}
  \begin{claim}\label{cl:51-H} 
    For any vertex $v$ of $G$, $v$ belongs to at most
    $n_i\Delta^{\gamma(n_i - 2) -\frac{m_i-m}{m-1}}$ copies of
    $H_i=(V_1,V_2,E)$ in $G$ that do not contain any special set in the
    images of $V_1$ nor in the image of $V_2$.  (That is $n_i
    \Delta^{(\gamma)(n_i-2)}$ copies for $m_i=m$ and $o(\Delta^{\gamma(n_i
      - 2)})$, otherwise.)
  \end{claim}
  \begin{proof}
    Let us consider only the copies of $H_i$ where $v$ corresponds to a
    given vertex $u$ of $H_i$. Now orient $H_i$ acyclically so that $u$ is
    the unique sink, and let us denote by $u=u_1,\ldots ,u_{n_i}$ the
    vertices of $H_i$ in such a way that for any $1\le j \le n_i$ the
    out-neighborhood of $u_j$ corresponds to its neighbors with index
    lower than $j$. Note that $d^+(u_j)\ge 1$ for all $1<
    j\le n_i$, and that $m_i = \sum_{1< j\le n_i} d^+(u_j)$.
    Observe that once $u_1,\ldots,u_{j-1}$ are set, there are at most
    $\Delta^{d^+(u_j)-\gamma(d^+(u_j)-1)}$ choices for $u_j$. This
    comes from the fact that the out-neighborhood of $u_j$ is
    monochromatic and hence cannot be a special $d^+(u_j)$-set.
    This leads to the following bound on the number of such copies of $H_i$.
    \begin{eqnarray*}
      \prod_{1<j\le n_i}\Delta^{d^+(u_j) - \gamma(d^+(u_j) -1)} & \le &\Delta^{m_i - \gamma(m_i -n_i + 1)} \\
                                                                & \le &\Delta^{(1-\gamma)m_i + \gamma(n_i - 1)} \\
                                                                & \le &\Delta^{\frac{-m_i}{m-1} + \gamma(n_i - 1)} \\
                                                                & \le &\Delta^{\frac{m-m_i}{m-1}-\gamma +\gamma(n_i - 1)}\\
    \end{eqnarray*}
    As there are $n_i$ possible choices for mapping $v$ in $H_i$, this
    concludes the claim.
  \end{proof}

  Now it is clear that an allowed coloring is a $(2,H_i)$-subgraph
  coloring for any $H_i\in \mathcal{F}$. An allowed coloring is thus a
  $(2,\mathcal{F})$-subgraph coloring.
  \begin{itemize}
  \item $\USBE_{Hi}(v,\ovphi,S)$ outputs $n_i-2$ vertices of
    $S$ including $v$ and such that the two remaining vertices, say
    $v_1$ and $v_2$, are such that $v_j\in V_j$ for $j=1,2$. Note that
    $s_{Hi}=n_i-2$. Again by Lemma~\ref{lem-USBE}, this
    function fulfills all the requirements.

  \item $\RBE_{Hi}(v,X,S,\varphi')$ outputs the partial
    coloring obtained from $\varphi'$ by properly extending the
    2-coloring of the two colored vertices of $S$ to the whole $S$.
  \end{itemize}
  Consider now
  \begin{eqnarray*}
    Q(x) & = & 1+ C_E\cdot x^{s_E}+\sum_{2\le j <n} C_{j\cdot {\rm
               Set}} \cdot x^{s_{j\cdot {\rm Set}}} + C_{\mathcal{F}_v^{>m}}
               \cdot x^{s_{\mathcal{F}_v^{>m}}} +
               \sum_{H_i\in\mathcal{F}_v^{\le m}}C_{Hi} \cdot x^{s_{Hi}}\\
         & = & 1 + \Delta x + \sum_{2\le j < n}
               \frac{1}{(j-1)!}(\Delta^{\gamma} x)^{j-1} + (m+1)4^{m+1}\Delta^{m}
               x^{m-1} \\ & & + \sum_{H_i \in \mathcal{F}_v^{\le m}} n_i \Delta^{\gamma(n_i - 2) -\frac{m_i-m}{m-1}} x^{n_i-2}\\
         & < & \Delta x + e^{\Delta^{\gamma} x} + 16(m+1) (4\Delta^\gamma x)^{m-1} + \sum_{H_i \in \mathcal{F}_v^{\le m}} n_i \left( \Delta^\gamma x\right)^{n_i - 2} \Delta^{-\frac{m_i-m}{m-1}}\\
  \end{eqnarray*}
  By setting $X= \frac{1}{4\Delta^{\gamma}}$, as
  $\Delta^\frac{-1}{m-1}<1$ and as for $H_i \in \mathcal{F}_v^{\le m}$ we have $3\le n_i
  \le m$, one obtains that
  \begin{eqnarray*}
    \frac{Q(X)}{X} & < & 4\Delta^{\gamma} \left( \frac14 + e^\frac14 + 16(m+1) + \frac14  k_v^{\le m} \cdot m \right)\\
  \end{eqnarray*}
  By Theorem~\ref{th:main-final}, $G$ admits an allowed coloring
  (hence a $(2,\mathcal{F})$-subgraph coloring) with $\ceil{Q(X)/X} <
  (k_v^{\le m}+71)(m+1)\Delta^{\gamma}$ colors. This concludes the
  proof of the first statement of the theorem.

  For the second statement we proceed similarly but there are two differences.
  \begin{enumerate}[(1)]
  \item Recall the partition of $\mathcal{F}$ into
    $\mathcal{F}_e^m$ and $\mathcal{F}_e^{>m}$ according to the number
    of edges. We replace the set $\mathcal{F}_e^{>m}$ by the set
    $\mathcal{T}_e^{m+1}$ of all trees on exactly $m+1$ edges. As every
    graph in $\mathcal{F}_e^{>m}$ contains a $(m+1)$-edge tree, a
    $(2,\mathcal{F}_e^m\cup \mathcal{T}_e^{m+1})$-subgraph coloring is
    also a $(2,\mathcal{F})$-subgraph coloring.
  \item All the graphs $\mathcal{F}_e^m\cup\mathcal{T}_e^{m+1}$
    are treated similarly by assigning each of them a specific bad
    event. There is no more the bad event type $\mathcal{F}_v^{>m}$.
  \end{enumerate}
  This yields to the following $Q(x)$.
  \begin{eqnarray*}
    Q(x) & = & 1+ C_E\cdot x^{s_E}+\sum_{2\le j <n}
               C_{j\cdot {\rm Set}} \cdot x^{s_{j\cdot {\rm Set}}} +
               \sum_{H_i\in\mathcal{F}_e^{m}\cup\mathcal{T}_e^{m+1}}C_{Hi} \cdot x^{s_{Hi}}\\ 
         & = & 1 + \Delta x + \sum_{2\le j < n}
               \frac{1}{(j-1)!}(\Delta^{\gamma} x)^{j-1} + \sum_{H_i \in \mathcal{F}_e^{m}\cup\mathcal{T}_e^{m+1}}
               n_i \left( x\Delta^\frac{m_i}{m_i-1} \right)^{n_i - 2}\\ 
         & < & \Delta x + e^{\Delta^{\gamma} x} + \sum_{H_i
               \in \mathcal{F}_e^{m}\cup\mathcal{T}_e^{m+1}} n_i \left( x\Delta^\frac{m_i}{m_i-1} \right)^{n_i -
               2}\\ 
         & < & \Delta x + e^{\Delta^{\gamma} x} + \sum_{H_i
               \in \mathcal{F}_e^{m}} n_i \left( \Delta^\gamma x\right)^{n_i - 2}+ \sum_{H_i
               \in \mathcal{T}_e^{m+1}} n_i \left( \Delta^\gamma x\right)^{n_i - 2}
               \Delta^{\frac{-1}{m-1}}\\
  \end{eqnarray*}
  By setting $X= \frac{1}{\Delta^{\gamma}}$ and as $3\le n_i \le m_i +1$, one obtains that
  \begin{eqnarray*}
    \frac{Q(X)}{X} & < & \Delta^{\gamma} \left( \Delta^{\frac{-1}{m-1}} + e + k_e^{m}(m+1) + |\mathcal{T}_e^{m+1}| \cdot (m+2) \Delta^\frac{-1}{m-1} \right)\\
                   & < & \Delta^{\gamma} \left( k_e^{m}(m+1) + e + o(1) \right)\\
  \end{eqnarray*}
  By Theorem~\ref{th:main-final}, $G$ admits an allowed coloring (hence
  a facial non-repetitive coloring) with $\ceil{Q(X)/X} < \left( k_e^{m} + 1 +
    o(1) \right) (m+1)\Delta^{\gamma}$ colors. This concludes the proof of
  the second statement of the theorem.
\end{proof}

\begin{remark}
  For given instances of $\mathcal{F}$, tighter bounds can be inferred
  with the general method. For example for star colorings of graphs,
  which correspond to $(2,\{P_4\})$-subgraph coloring, it is not
  necessary to have bad events for special sets. It suffice to have
  one bad event ensuring that the coloring is proper (with $C_1=\Delta$
  and $s_1=1$), and one bad event to avoid bicolored $P_4$'s (with
  $C_2=2\Delta(\Delta-1)^2$ and $s_2=2$). This yields to the bound
  $2\sqrt{2}\Delta^\frac32 + \Delta -\sqrt{8\Delta}+1$ (by setting
  $X=1/(\sqrt{2\Delta}(\Delta-1))$), similar to the one in~\cite{EP12}.
\end{remark}

% Local variables:
% mode: latex
% TeX-master: "0_main.tex"
% End:
  %% section 5
\section{Conclusion}\label{sec:ccl}

One should note that the framework presented in
Section~\ref{sec:general_la_vraie} may, in some cases, benefit from
some sophistication. The version we presented here seems to be a good
compromise between efficiency and clarity for the applications we
considered.  We have seen in Subsection~\ref{subsec:thue-edge} how, at
any step $i$, one can get benefit from $\ovphi_{i-1}$ to decrease the
values $C_j$. One could also take into account the order in which the
vertices of $\ovphi_{i-1}$ have been colored. For example, if $(u,v)$
is a special pair (as in Subsection~\ref{ssec:acyclic}) and $u$ has
been colored after $v$ to obtain $\varphi_{i-1}$, then one could be sure
that the colors of $u$ and $v$ are distinct. Thus one would not have
to consider bad events where $u$ and $v$ are colored the same. One
could thus imagine that all the functions presented in
Subsection~\ref{subsec:requ-fun} could depend on the ordering $\pi$ in which
the vertices of $\ovphi_{i-1}$ were colored.

Finally an interesting way of improving this framework would be
handling algorithms where the costs of a given bad event may vary. For
example, one can imagine that, for some vertices, a type $j$ bad event
costs $C_j$, while for some other vertices the cost is $C'_j$. A
simple way to analyze this is to set the cost of each type $j$ bad
event to $\max\{C_j,C'_j\}$. We wonder whether there exists a better
approach.

% Local variables:
% mode: latex
% TeX-master: "0_main.tex"
% End:

  %% Conclusion

\newpage
\appendix

\section{The smooth implicit-function schema}\label{sec:meir_moon}

In Section~\ref{sec:general}, we prove Theorem~\ref{thm:nb_of_rec} by
using a machinery provided by a theorem of Meir and
Moon~\cite{MM89} (see Theorem~\ref{thm:meir_moon}) on the
singular behaviour of generating functions defined by a \emph{smooth
  implicit-function schema}.

\begin{definition}[Smooth implicit-function schema {\cite[Definition~VII.4, p.~467]{FS08}}]\label{def:SIFS}
  Let $A(y)$ be a function analytic at~0, $A(y)=\sum_{t\ge 0}a_ty^t$,
  with $a_0=0$ and $a_t\geq 0$. The function is said to belong to the
  \emph{smooth implicit-function schema} if there exists a bivariate
  function $G(y,z)$ such that $A(y) = G(y,A(y))$, where $G(y,z)$ satisfy the following conditions:
  \begin{enumerate}[$(a)$]
  \item $G(y,z)=\sum_{m,n\ge 0}g_{m,n}y^mz^n$ is analytic in a domain $|y| < R$ and $|z| < S$, for some $R,S>0$.
  \item The coefficients of $G$ satisfy 
    \begin{eqnarray*}
      g_{m,n}\ge 0,~g_{0,0}=0,~ g_{0,1}\neq 1,\\
      g_{m,n}>0 \; \mbox{for some $m\geq 0$ and some $n \geq 2$}.      
    \end{eqnarray*}
  \item There exist two numbers $r$ and $s$, such that $0 < r < R$ and $0 < s < S$, satisfying the system of equations\footnote{ $G_y$ (resp. $G_z$) denotes the derivative of $G$ with respect to its first (resp. second) variable.}
    $$G(r, s)=s,~G_z(r, s)=1,\quad \mbox{with} \;\; r<R,~s<S$$
    which is called the \emph{characteristic system}.
  \end{enumerate}
\end{definition}

\begin{theorem}[Meir and Moon~\cite{MM89},{\cite[Theorem~VII.3, p. 468]{FS08}}]\label{thm:meir_moon}
  Let $A(y)$ belong to the smooth implicit-function schema defined by
  $G(y,z)$ with $(r,s)$ the positive solution of the characteristic
  system. Then, $A(y)$ converges at $y=r$, where it has a square-root singularity,
  $$\lim_{y\to r}A(y)=s-\gamma\sqrt{1-\frac y r} + O\paren{1-\frac y r}, \quad \mbox{with} \;\; \gamma=\sqrt{\frac{2rG_y1(r,s)}{G_{zz}(r,s)}},$$
  the expansion being valid in a $\Delta$-domain. In addition, if
  $A(y)$ is aperiodic, then $r$ is the unique dominant singularity of $A$ and the coefficient satisfy
  $$\lim_{t\to \infty}[y^t]A(y) = \frac\gamma{2\sqrt{\pi t^3}}r^{-t}\paren{1+O(t^{-1})}.$$

\end{theorem}

%%% Local Variables:
%%% mode: latex
%%% TeX-master: "0_main.tex"
%%% End:
 %% section A

\end{document}